\documentclass[a4paper,preprintnumbers,floatfix,superscriptaddress,aps,10pt,twocolumn,notitlepage,longbibliography,noarxiv]{revtex4-1}
\usepackage[utf8]{inputenc}
\usepackage[english]{babel}
\usepackage[T1]{fontenc}
\usepackage{amsmath}
\usepackage{amssymb}
\usepackage{amsthm}
\usepackage{mathtools}
\usepackage[colorlinks=true,citecolor=blue,urlcolor=magenta,linkcolor=red]{hyperref}
\usepackage{todonotes}
\setuptodonotes{inline}
\usepackage{scalerel}
\usepackage{braket}

\usepackage[caption=false]{subfig}
\usepackage{accents}
\usepackage{bm,bbm}
\usepackage{lipsum}
\usepackage{nameref}
\usepackage[capitalize]{cleveref}
\crefname{methods}{Methods}{Methods}
\usepackage{array} 
\usepackage{mdframed}
\usepackage{comment}
\usepackage{bm}
\usepackage{listings}
\usepackage{capt-of} 
\usepackage{graphicx}
\usepackage{xcolor}
\usepackage{booktabs}
\usepackage{tabularx}
\usepackage{multirow}
\definecolor{codegreen}{rgb}{0,0.6,0}
\definecolor{codegray}{rgb}{0.5,0.5,0.5}
\definecolor{codepurple}{rgb}{0.58,0,0.82}
\definecolor{backcolour}{rgb}{0.95,0.95,0.92}

\usepackage{tikz}
\usepackage[outline]{contour}

\usetikzlibrary{calc}

\lstdefinestyle{mystyle}{
    backgroundcolor=\color{backcolour},   
    commentstyle=\color{codegreen},
    keywordstyle=\color{magenta},
    numberstyle=\tiny\color{codegray},
    stringstyle=\color{codepurple},
    basicstyle=\ttfamily\footnotesize,
    breakatwhitespace=false,         
    breaklines=true,                 
    captionpos=b,                    
    keepspaces=true,                 
    numbers=left,                    
    numbersep=5pt,                  
    showspaces=false,                
    showstringspaces=false,
    showtabs=false,                  
    tabsize=2
}

\DeclarePairedDelimiterX\ketbra[2]{\vert}{\vert}%
  {#1\kern0.15ex\delimsize\rangle\delimsize\langle\kern0.15ex\mathopen{}#2}
\newcommand{\kb}[2]{\ketbra{#1}{#2}}

\newcommand{\1}{\openone}
\newcommand{\CC}{\mathbb{C}}
\newcommand{\RR}{\mathbb{R}}
\newcommand{\ZZ}{\mathbb{Z}} 
\newcommand{\Znn}{\mathbb{Z}_{\geq 0}} 
\newcommand{\Rnn}{\mathbb{R}_{\geq 0}} 

\newcommand{\LL}{\mathcal{L}}
\renewcommand{\H}{\mathcal{H}}

\newcommand{\U}{\mathrm{U}} 
\newcommand{\OO}{\mathrm{O}} 
\newcommand{\ii}{\ensuremath\mathrm{i}}
\newcommand{\e}{\ensuremath\mathrm{e}} 
\renewcommand{\i}{\ensuremath\mathrm{i}} 

\newcommand{\veps}{\varepsilon}

\renewcommand{\Pr}{\mathbb{P}} 
\newcommand{\T}{\intercal} 

\newcommand{\avg}{{\mathrm{avg}}}
\newcommand{\Fid}{\mathrm{F}}

\let\Re\undefined
\DeclareMathOperator{\Re}{\mathrm{Re}} 
\let\Im\undefined
\DeclareMathOperator{\Im}{\mathrm{Im}} 

\newcommand{\Zgate}{\mathrm{Z}}
\newcommand{\Z}{\mathrm{Z}}

\newcommand{\Xgate}{\mathrm{X}}
\newcommand{\X}{\Xgate}
\newcommand{\Ygate}{\mathrm{Y}}
\newcommand{\Y}{\mathrm{Y}}

\let\Set\undefined

\DeclarePairedDelimiterX\Set[1]\{\}{%
  
  #1
}
\DeclarePairedDelimiterX{\abs}[1]{\lvert}{\rvert}{%
  \ifblank{#1}{\,\cdot\,}{#1}
} 

\DeclarePairedDelimiterX\norm[1]\lVert\rVert{%
  \ifblank{#1}{\,\cdot\,}{#1}
}   

\renewcommand{\vec}[1]{\boldsymbol{\mathbf{#1}}}

\makeatletter
\newtheorem*{rep@theorem}{\rep@title}
\newcommand{\newreptheorem}[2]{%
\newenvironment{rep#1}[1]{%
 \def\rep@title{#2 \ref{##1}}%
 \begin{rep@theorem}}%
 {\end{rep@theorem}}}
\makeatother
\makeatletter
\newtheorem*{rep@lemma}{\rep@title}
\newcommand{\newreplemma}[2]{%
\newenvironment{rep#1}[1]{%
 \def\rep@title{#2 \ref{##1}}%
 \begin{rep@lemma}}%
 {\end{rep@lemma}}}
\makeatother
\usepackage{thmtools}
\usepackage{thm-restate}
\newtheorem{theorem}{Theorem}
\newreptheorem{theorem}{Theorem}

\newtheorem{definition}{Definition}
\newtheorem{observation}[theorem]{Observation}

\newtheorem{lemma}[theorem]{Lemma}
\newreptheorem{lemma}{Lemma}

\newenvironment{proofsketch}[1]{%
  \renewcommand{\proofname}{Proof sketch of {#1}}\proof}{\endproof}
  
\crefname{figure}{Fig.}{Figs.}
\crefname{algorithm}{Protocol}{Protocols}
\crefname{definition}{Def.}{Defs.}

\newcounter{protocol}
\renewcommand{\theprotocol}{\arabic{protocol}}

\definecolor{jan}{rgb}{0.4, 0, 1}

\definecolor{mari}{rgb}{1.0, 0.0, 0.22}

\usepackage{acronym}
\usepackage{algorithm}
\usepackage[noend]{algorithmic}
\floatname{algorithm}{Protocol}

\newacro{RB}{randomized benchmarking}
\newacro{GST}{gate set tomography}
\newacro{POVM}{positive operator-valued measure}
\newacro{PVM}{projection-valued measure}
\newacro{CP}{completely positive}
\newacro{CPTP}{completely positive trace preserving}
\newacro{PSD}{positive semidefinite}
\newacro{NISQ}{noisy intermediate-scale quantum}
\newacro{SPAM}{state preparation and measurement} 
\newacro{ONB}{orthonormal basis}
\newacro{SDI}{semi-device independent}
\newacro{QSQ}{quantum system quizzing}
\newacro{MUBs}{mutually unbiased bases}
\newacro{KKT}{Karush–Kuhn–Tucker}
\newacro{QCVV}{Quantum Characterization, Verification, and Validation}
\newacro{PTM}{Pauli transfer matrix}
\newacroplural{PTM}[PTMs]{Pauli transfer matrices}

\DeclareMathOperator{\diag}{diag}

\DeclareMathOperator{\Tr}{Tr} 

\DeclareMathOperator{\Bin}{Bin}

\begin{document}

\title{Sound and Efficient Certification of High-Quality Qubit Operations: \\
Theory and Experiment}

\author{Nikolai Miklin}
\affiliation{Institute for Quantum-Inspired and Quantum Optimization, Hamburg University of Technology, Germany}
\affiliation{Institute for Applied Physics, Technical University of Darmstadt, Darmstadt, Germany}

\author{Jan N\"oller}
\affiliation{Department of Computer Science, Technical University of Darmstadt, Darmstadt, Germany}

\author{Jos\'e Mart\'inez}
\affiliation{Department of Computer Science, Technical University of Darmstadt, Darmstadt, Germany}
\affiliation{QUANTUM, University of Mainz, Department of Physics, Staudingerweg 7, Germany}

\author{Lucas B. Vieira}
\affiliation{Department of Computer Science, Technical University of Darmstadt, Darmstadt, Germany}

\author{Ulrich Poschinger}
\affiliation{QUANTUM, University of Mainz, Department of Physics, Staudingerweg 7, Germany}

\author{Ferdinand Schmidt-Kaler}
\affiliation{QUANTUM, University of Mainz, Department of Physics, Staudingerweg 7, Germany}

\author{Mariami Gachechiladze}
\affiliation{Department of Computer Science, Technical University of Darmstadt, Darmstadt, Germany}

\begin{abstract}
Can a high-quality quantum gate be certified when uncharacterized state-preparation and measurement errors are dominant? Can this be achieved with low experimental overhead? Here we introduce a sound black-box certification protocol for a single-qubit gate based on a small set of fixed, deterministic sequences. From the data, the protocol derives finite-sample bounds on the gate's rotation eigenvalue, a gauge-invariant property. Its phase reveals the accuracy of the rotation angle, while its modulus quantifies the loss of coherence under repeated gate applications. We implement the protocol on a $^{40}\mathrm{Ca}^{+}$ trapped-ion processor and certify the $\sqrt{\Xgate}$-gate rotation eigenvalue using $22\,000$ circuit executions, and demonstrate the robustness of certification to state-preparation and measurement errors by deliberately degrading the readout.
Finally, we prove that these spectral constraints imply, up to a physically meaningful unitary change of basis, a rigorous average gate-fidelity lower bound for every time-independent qubit model compatible with the data. In both readout settings, the spectral bounds yield the same fidelity certificate of $99.94(3)\%$ with $99\%$ confidence. Our results establish a new standard for quantum-gate certification by combining soundness and experimental efficiency without requiring trusted reference operations, randomized circuits, or model fitting.
\end{abstract}

\maketitle
\makeatletter
\hypersetup{pdftitle = {Sound and Efficient Certification of High-Quality Qubit Operations: \\
Theory and Experiment},
       pdfauthor = {Nikolai Miklin, Jan Nöller, José Martínez, Lucas B. Vieira, Ulrich
Poschinger, Ferdinand Schmidt-Kaler, Mariami Gachechiladze},
       pdfsubject = {Quantum computing},
       pdfkeywords = { 
              quantum, certification, benchmarking, verification, quizzing, dimension, assumption, 
              self-testing, soundness, sound,
              verifier, user, remote, computer, computation, 
              dynamics, gate, circuit, algorithm, protocol, 
              gauge, freedom, 
              }
      }

The arrival of high-fidelity quantum gates~\cite{Ballance2016highfidelity,Evered2023high,gaebler2016high,moses2023racetrack} changes the role of certification. At this level, the experimental question is not merely whether a device can perform a gate well, but whether such performance can be verified with comparable precision~\cite{Eisert2020review,kliesch2021theory,Nielsen2021gatesettomography,blume2017demonstration}. This shift exposes a fundamental bottleneck: certifying an infidelity on the order of e.g., $10^{-4}$ requires resolving a very small deviation from ideal behavior, which in many existing methods relying on direct statistical estimation, can demand on the order of $10^{8}$ measurements~\cite{Flammia2011direct, Eisert2020review,kliesch2021theory}.

This statistical barrier is compounded by a second, equally important difficulty. In experiments, the gate of interest is never accessed in isolation, but only through imperfect~\ac{SPAM}. These \ac{SPAM} errors are often larger than the gate errors one aims to certify~\cite{arute2019quantum,Evered2023high,aasen2024readout,wu2021strong,moses2023racetrack,marxer2026above,ding2025high}. Consequently, a certification protocol that is not robust to \ac{SPAM} errors can confuse them with errors of the gate itself, or conversely, overstate the quality of a gate by absorbing imperfections into an inaccurate model of the experiment~\cite{Merkel2013self,Proctor2017what,qi2019comparing,helsen2022general,Nielsen2021gatesettomography,Blume2013robust,smith2025single,brieger2023compressive,miller2026scalable,Emerson2005scalable,knill2008randomized}. More generally, a meaningful certification procedure must be \textit{sound}: passing the test should be possible only for genuinely high-quality gates~\cite{kliesch2021theory,noller2025classical,noller2025sound,schroeder2025certifying,liu2020efficient,van2000self,Pallister2018optimal}. 

This motivates the central question of this work: can a high-quality quantum gate be certified with soundness guarantees, robustness to \ac{SPAM} errors, and a practical sample complexity? We answer this question affirmatively. In a black-box setting, without prior characterization or trusted reference operations, we certify that a gate operates in a high-fidelity regime even when other components of the apparatus, including state preparation and measurement, are substantially noisier.

The key step is to experimentally identify and bound gauge-invariant spectral properties of the gate and translate them into a certificate for its average gate fidelity. This avoids model-dependent process reconstruction, gauge optimization, and decay-curve fitting~\cite{Rudnicki2018gauge,Lin2019freedom,DiMatteo2020operationalgauge,Helsen2019Spectral,Nielsen2021gatesettomography,Chen2023learnability}. In contrast to randomized benchmarking and related approaches, the protocol requires no randomized circuits, group averaging, or twirling~\cite{Emerson2005scalable,knill2008randomized,Magesan2011scalable,Magesan2012characterizing,magesan2012efficient,Wallman2016noise,Proctor2017what,Erhard2019characterizing,Proctor2019direct,helsen2022general,hashim2020randomized}. It uses only a small set of fixed, deterministic gate sequences and has low experimental and computational overhead.

We develop the theory framework for single-qubit $\pi/2$-pulse gates, elementary operations that underpin Ramsey interferometry, precision spectroscopy~\cite{cronin2009optics,muller2008atom,ramsey1950molecular,carr1954effects,ludlow2015optical,hahn1950spin,viola1999dynamical}, quantum control and circuit compilation~\cite{Barenco1995elementary,kelly2014optimal,khatri2019quantum,chen2023compiling}, and enables certification of more general gate sets~\cite{noller2025classical,noller2025sound}.
We experimentally validate the developed theoretical framework on a \(^{40}\mathrm{Ca}^{+}\) trapped-ion quantum processor~\cite{ruster2016long,fsk2016longlived,fsk2020shuttling,kaufmann2017fast}. Using only the order of $10^{4}$ circuit executions, we certify the gate's rotation eigenvalue and thereby establish a fidelity lower bound above $99.9\%$. Numerical simulations across a range of noise models further demonstrate that rigorous certification remains practical at the precision frontier of current quantum hardware.

The remainder of this Article is organized as follows. We first formulate the certification task and develop a rigorous two-stage method for deriving gate-quality certificates from experimental data. A minimal preliminary test determines whether the device operates in the high-quality regime, after which a refined protocol provides a higher-resolution certificate. We then present the experimental implementation and results, complemented by numerical simulations across a broad range of noise models. Finally, we provide a software package for readily deploying the method on other experimental platforms. 

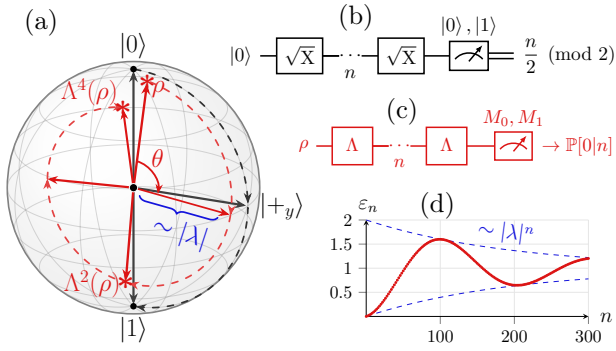
\begin{figure}[!t]
    \centering
    \begin{tikzpicture}[
        panel/.style={anchor=north west,inner sep=0pt,outer sep=0pt,align=center},
        sublabel/.style={anchor=north west, font=\normalsize}]
        \node[panel,minimum width=0.5\linewidth] (bloch) at (0,0) {\includegraphics[width=0.5\linewidth]{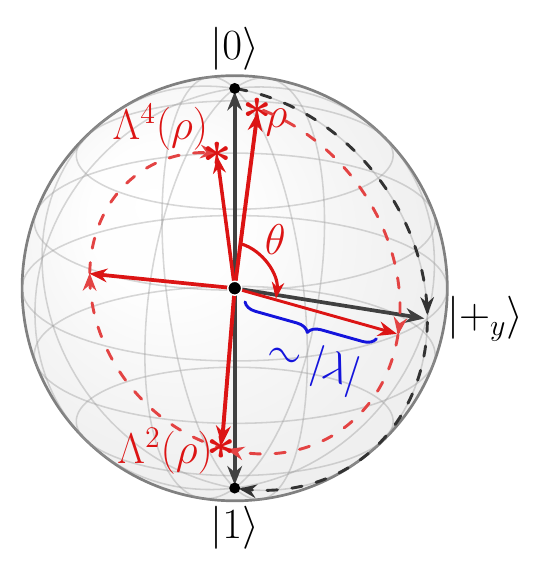}};
        
        \node[panel,minimum width=0.61\linewidth] (ideal) at ([xshift=0.35\linewidth]bloch.north west) {\includegraphics[width=0.61\linewidth]{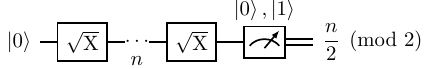}};

        \node[panel,minimum width=0.55\linewidth] (noisy) at ([yshift=-4mm,xshift=7mm]ideal.south west) {\includegraphics[width=0.5\linewidth]{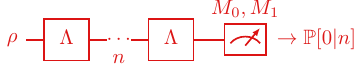}};

        \node[panel,minimum width=0.55\linewidth] (signal) at ([yshift=-3mm,xshift=2mm]noisy.south west) {\includegraphics[width=0.45\linewidth]{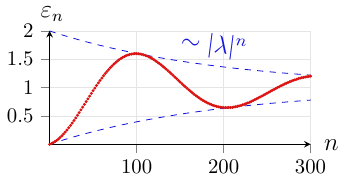}};

        \node[sublabel,xshift=8pt,yshift=6pt] at (bloch.north west) {(a)};
        \node[sublabel,xshift=40pt,yshift=8pt] at (ideal.north west) {(b)};
        \node[sublabel,xshift=38pt,yshift=8pt] at (noisy.north west) {(c)};
        \node[sublabel,xshift=45pt,yshift=4pt] at (signal.north west) {(d)};
    \end{tikzpicture}

\caption{\textbf{Schematics of the certification protocol}. (a) Initial state $\rho$ and evolved states $\Lambda^2(\rho)$ and $\Lambda^4(\rho)$ for ideal (black) and noisy (red) implementations, represented on the Bloch sphere. (b) Target circuit comprising an even number $n$ of applications of the $\sqrt{\Xgate}$ gate. (c) Corresponding general qubit model, in which all components of the experiment are fully uncharacterized. (d) Typical noisy signal, showing $\veps_n = 1-(-1)^{n/2}\left(\Pr[0\vert n]-\Pr[0\vert n+2]\right)$ as a function of $n$. The dashed envelopes indicate the decay of the oscillations, which arises from the contraction of the Bloch vector under the noisy channel shown in (a), whereas the oscillation frequency reflects the deviation of $\theta$ from the target rotation angle $\pi/2$.}
    \label{fig:scheme}
\end{figure}

\subsection*{Certification task and gauge invariant parameters}
A typical quantum-computing experiment comprises state preparation, evolution through a sequence of gates, and measurement. In a certification task, target descriptions of these components are specified, and experimental data is used to determine how closely their implementations match the targets. However, \ac{SPAM} errors can systematically bias estimates of gate quality. We therefore assume no prior knowledge of the state preparation or measurement. Instead, the state, measurement, and gate are all treated as unknown and certified jointly within a single, self-consistent experiment.

Here, we focus on certifying a $\pi/2$-pulse gate. In the reference frame defined by the initial state $\ket{0}$ and the measurement basis $\{\ketbra{0}{0},\ketbra{1}{1}\}$, the target gate is
\begin{equation}\label{eq:sqrtX}
    \sqrt{\Xgate}=\ketbra{+}{+}+\i\ketbra{-}{-},
\end{equation}
where $\ket{\pm}=\frac{1}{\sqrt{2}}(\ket{0}\pm\ket{1})$. Together, the initial state, gate, and measurement constitute the \emph{target model} shown in \cref{fig:scheme}~(b). By contrast, we assume that the implemented device is described by an unknown qubit initial state $\rho$, an unknown two-outcome measurement $\{M_0,M_1\}$, and an unknown, time-independent, quantum channel $\Lambda$, as shown in \cref{fig:scheme}~(c). Our goal is to derive, solely from the observed statistics, rigorous guarantees on the quality of the implemented model relative to this target. 

A central difficulty is that experimental data determines probabilities, not a unique mathematical description of a device. Simultaneous transformations of the state, gate, and measurement can leave all observed probabilities unchanged; these observationally indistinguishable reparametrizations are commonly called \emph{gauge transformations}~\cite{Merkel2013self,Proctor2017what}. Some are purely algebraic and need not preserve the physical interpretation of the individual model components. By contrast, the physically relevant ambiguity is the choice of quantum reference frame: by Wigner's theorem, transformations preserving quantum transition probabilities are unitary or antiunitary~\cite{wigner1931gruppentheorie}. Applying the same such transformation to all components amounts to an experimentally undetectable change of basis. We denote the corresponding transformation of the target gate by $\sqrt{\Xgate}_{\vert U}\coloneqq U\sqrt{\Xgate}U^\dagger$ (with the complex conjugation applied to the right-hand side in the case of an anti-unitary transformation).

Our goal is to identify physical quantities that are directly determined by the experimental data and invariant under its full gauge redundancy, and to use them to estimate the gate quality relative to the target up to the unavoidable unitary or antiunitary change of basis. This ensures that nonphysical gauge choices cannot transfer errors between the gate and the \ac{SPAM} components or artificially inflate the certified fidelity (see~\cref{app:gauge} for an immediate example of such overestimation using an invertible gauge transformation).

For the ideal $\pi/2$-pulse gate, the relevant gauge-invariant information has a simple geometric interpretation~\cite{blumekohout2025qcvv}. The gate in  \cref{eq:sqrtX} implements a quarter-turn rotation about the $\Xgate$ axis of the Bloch sphere, as illustrated in \cref{fig:scheme}~(a). A change of reference frame changes the rotation axis but leaves the rotation angle invariant.
A noisy, near-ideal implementation $\Lambda$ can be viewed as an imperfect rotation accompanied by a contraction or deformation of the Bloch sphere. Because $\Lambda$ need not be unitary, the rotation angle alone is insufficient to characterize it. In the \ac{PTM} representation of qubit channels~\cite{greenbaum2015introduction}, the eigenvalues of the \ac{PTM} provide a natural, although not complete, set of gauge-invariant parameters. We therefore focus on them.

\section*{Results}
We begin by showing that, in the sufficiently low-noise regime, the quality of an implemented $\pi/2$-pulse gate can be certified from a single complex, gauge-invariant parameter. This parameter is one of the complex-conjugate pair of eigenvalues of the \ac{PTM} of the implemented channel $\Lambda$. We denote it by $\lambda$ and refer to it as the \emph{rotation eigenvalue}.
Indeed, knowing $\lambda$, we obtain a lower bound on the average gate fidelity
\begin{equation}\label{eq:fid}
\Fid_\avg(\sqrt{\Xgate}_{\vert U},\Lambda) \geq
\frac{1}{3}+\frac{1}{3}\left(\abs{\lambda}+\Im[\lambda]\right),
\end{equation}
with respect to the ideal $\pi/2$ rotation in \emph{some} unitary gauge $U$.
The two terms in this bound have a direct operational interpretation. The quantity $\Im[\lambda]$ captures the coherent part of the rotation and, therefore, characterizes angle errors, while $\abs{\lambda}$ captures the contraction of the rotating plane or, in other words, characterizes the incoherent decay. The derivation of the bound is technically straightforward but notation-heavy. We therefore defer the details to \cref{methods:QSQ}. The central challenge is to certify $\lambda$ rigorously from experimental statistics in a black-box setting, despite unknown~\ac{SPAM} errors.

In the above, we glossed over the fact that the bound in \cref{eq:fid} applies if the implementation is not too noisy.
However, how can one be sure of that in the black-box scenario?
To that end, we develop a two-stage certification procedure: a preliminary test certifies that the experiment lies in the high-quality regime;
while a refined stage produces a certificate for the rotation eigenvalue $\lambda$ in the \ac{SPAM}-robust manner, which, via \cref{eq:fid}, translates into the fidelity lower bound.

\subsection*{Certification from deterministic sequences}
Our method is based on querying the uncharacterized device by repeated application of the tested gate to the same initial state and making the same measurement at the end (see \cref{fig:scheme} (b,c)).
The tested gate sequences are chosen such that the ideal implementation produces a deterministic outcome $0$ or $1$ of the measurement.
For the $\pi/2$-pulse gate applied to the basis state $\ket{0}$, these deterministic sequences are exactly those of even length. 
More precisely, for even $n$, the outcome of the measurement in the circuit in \cref{fig:scheme} (b) is $n/2 \pmod{2}$.
For an implemented model in \cref{fig:scheme}~(c), the probability of obtaining outcome $0$ after $n$ applications of the gate is
\begin{equation}\label{eq:probs}
\Pr[0\vert n]=\Tr[M_0\Lambda^n(\rho)].
\end{equation}
In our analysis, we pair the error probabilities:
\begin{equation}\label{eq:epsilon_n}
    \veps_n\coloneqq 1-(-1)^\frac{n}{2}\left(\Pr[0\vert n]-\Pr[0\vert n+2]\right).
\end{equation}
This has the effect of reducing the readout noise (see \cref{app:PTM} for a discussion).
In the experiment, we estimate the probabilities $\Pr[0\vert n]$ in a straightforward manner by executing the same circuit $N$ times, from which we obtain the estimates of $\veps_n$ (see \cref{methods:statistics} for more details on statistical analysis).

\subsection*{Short-sequence test}
We first describe the short-sequence test.  Remarkably, testing only sequences of lengths $n\in\Set{0,2,4}$ suffices to obtain non-trivial certificates on the quality of the gate and the~\ac{SPAM}.
We prove that there exists a unitary gauge $U$ such that the average gate fidelity satisfies
\begin{equation}\label{eq:fid_short}
    \Fid_\avg(\sqrt{\Xgate}_{\vert U},\Lambda) \geq \frac{1}{3}+\frac{2}{3}
    \sqrt{1-(\sqrt{\veps_0}+\sqrt{\veps_2})^2},
\end{equation}
and the state fidelity, as well as the total variation distance for the measurement, are bounded by $1-\veps_0$ and $\veps_0$, respectively (see \cref{methods:short} for a detailed statement).

This full-model certificate is useful both theoretically and experimentally, and we report the corresponding short-sequence fidelity bounds in \cref{tab:qsq_results}. 
However, the bound in \cref{eq:fid_short} is more sensitive to \ac{SPAM} errors than to gate errors.
As shown in \cref{eq:fid_short} and discussed below, deliberately increasing the readout error lowers the experimental bound. We thus use the bound in \cref{eq:fid_short} only to verify the high-quality regime required for the subsequent rotation-eigenvalue certificate.

More precisely, to apply our long-sequence tests, it is sufficient to ensure that $\sqrt{\veps_0}+\sqrt{\veps_2}<\frac{1}{2}$. 
To that end, one can execute a simple statistical test of checking whether any wrong outcomes occur in $N$ runs for each of the sequences $n\in\Set{0,2,4}$. Here, $N$ can be as low as $72$ for the confidence level $1-\delta = 0.99$ (see \cref{methods:statistics}).

\subsection*{Long-sequence tests}
Conditional on passing the preliminary test, the long-sequence protocol certifies how close $\lambda$ is to the target eigenvalue $\i$, which corresponds to the perfect $\pi/2$ rotation. As in existing methods~\cite{helsen2022general,Nielsen2021gatesettomography}, we mitigate the effect of \ac{SPAM} by executing longer sequences to amplify gate errors.
In the high-quality regime, the signal $\veps_n$ as a function of $n$ behaves like a damped oscillation, governed by the complex eigenvalue $\lambda\coloneqq\abs{\lambda}\e^{\i\theta}$ (see \cref{fig:scheme} (d)). The phase error $\abs{\theta-\pi/2}$ determines the oscillation frequency and captures coherent over- or under-rotation, while the modulus $\abs{\lambda}$ approximately determines the envelope and captures the incoherent decay.

A key feature of our long-sequence stage is that it relies on the estimation of $\veps_n$  for only a few values of $n$. Specifically, the phase $\theta$ can be certified using eight circuits, corresponding to four quantities $\veps_n$ with $n\in\{0,2j,4j,6j\}$ for a chosen odd integer $j$.
These four $\veps_n$ give the following analytic constraint
\begin{equation}\label{eq:phase_cert}
    \abs{1+\lambda^{2j}}\leq
    \sqrt{\frac{\veps_{6j}+\veps_{4j}-\veps_{2j}-\veps_{0}}{1-\veps_{2j}}}=:\xi.
\end{equation}
This constraint has two immediate consequences. First, it bounds the phase through
\begin{equation}\label{eq:bound_phi}
    \abs{\sin(2j\theta)}\leq \xi .
\end{equation}
Second, it gives an amplitude bound
\begin{equation}\label{eq:bound_lambda}
    \abs{\lambda}\geq (1-\xi)^{1/(2j)} .
\end{equation}
We defer the derivation of \cref{eq:phase_cert} to \cref{methods:long} and focus here on the intuition behind the two statements.

The integer $j$ acts as an amplification parameter. If $\theta=\pi/2+\Delta$, then, for odd $j$,
$ \abs{\sin(2j\theta)}=\abs{\sin(2j\Delta)}$.
The signal is therefore sensitive to the accumulated phase error $2j\Delta$.  
It is important to note that $j$ cannot be chosen completely arbitrarily. If $j$ is too large relative to the unknown oscillation period, the four stroboscopic points may skip over oscillation minima and become compatible with a slower, aliased phase. In practice, this ambiguity can be avoided by a cheap preliminary scan: one first samples a small number of various sequence lengths, with modest shot numbers, only to reveal the coarse damped-oscillation pattern. The value of $j$ is then chosen so that the four points used in \cref{eq:phase_cert} occur ideally before the first minimum in the signal (see Fig.~\ref{fig:scheme} (d) and numerical experiments in~\cref{methods:simulations}). Higher statistical precision is required only for these selected stroboscopic points.  As a result, the long-sequence data yields a certified interval for the rotation angle $\abs{\theta-\pi/2} \leq \frac{\arcsin(\xi)}{2j}$.

After the phase has been localized, the remaining task is to lower-bound $\abs{\lambda}$, which quantifies the decay of the rotating eigenmode. 
The bound in \cref{eq:bound_lambda} improves with increasing $j$, but has a disadvantage of not being sensitive enough to distinguish the effect of over-rotation and decoherence.
To circumvent this, we use one additional long-sequence point, which we denote as $\ell$. The coarse scan used to avoid phase aliasing also guides this choice: once it has revealed a damped oscillation, rather than a monotone loss of contrast, one can identify late sequence lengths where the oscillation is still visibly present.

The role of $\ell$ is not to resolve the full decay envelope, but to distinguish a surviving oscillatory signal from one that has already attenuated. If $\abs{\lambda}$ were substantially smaller than one, the contrast would decay rapidly and the paired signal would approach its attenuated value, $\veps_\ell\approx 1$. Conversely, observing $\abs{1-\veps_\ell}$ to remain clearly separated from zero at a large $\ell$ certifies that the rotating eigenmode has not decayed too quickly. 
More formally, we can prove that if $\abs{1-\veps_\ell}\geq 7/10$, $\abs{\lambda}\geq 9/10$ for $\ell\geq 22$, then
\begin{equation}\label{eq:amplitude_bound}
    \abs{\lambda}\geq \abs{1-\veps_\ell}^{1/(\ell-11)}.
\end{equation}  
The details about the derivation of \cref{eq:amplitude_bound} are given in \cref{methods:long}.
Notice here that \cref{eq:bound_lambda} can be used to assert the condition of the above bound. 

The certificates on the absolute value and the phase in \cref{eq:amplitude_bound} and \cref{eq:bound_phi}, respectively, allow us to prove the quality of the gate that is above that of \ac{SPAM}, as we demonstrate in the experimental results section.
At the same time, apart from the prior exploration of the signal $\veps_n$ in \cref{fig:scheme} (d) and preliminary test of high-quality regime, both of which can be done with very few samples, only an order of ten circuits needs to be sampled with a higher number of runs to achieve the required statistical precision (see \cref{methods:statistics} for details).

\subsection*{Experimental implementation}\label{sec:experimental_implementation}
We implement the proposed certification method on a trapped-ion shuttle-based \cite{fsk2020shuttling} quantum processor \cite{poschinger2009coherent}  with $^{40}\mathrm{Ca}^+$ ions. The relevant atomic energy levels and transitions are shown in \cref{methods:fig:experiment_figures}~(a) in~\cref{methods:experiment}. The qubit states are encoded in the Zeeman sub-levels of the ground-state $\ket{0} \equiv \ket{4S_{1/2},\; m_J = +1/2}$ and $\ket{1} \equiv \ket{4S_{1/2},\; m_J = -1/2}$, separated by approximately $2\pi\times10.5\,\text{MHz}$. The absence of hyperfine structure in the $^{40}\mathrm{Ca}^+$ Zeeman qubit fully mitigates any information leakage into parasitic states. We have shown spin-echo coherence times of up to 2.1(1)\,s  \cite{fsk2016longlived}, yielding a dephasing time ($T_2$) that is several orders of magnitude larger as compared to the duration of gate operations \cite{fsk2025compiler}. 

State preparation is performed in a two-stage laser-driven optical pumping sequence that depletes $\ket{1}$ and prepares $\ket{0}$ with an average resulting state fidelity better than $10^{-3}$ \cite{Hilder2022phd}. Single-qubit quantum gate operations are performed by two-photon stimulated Raman transition which is driven by a pair of co-propagating beams near 397\,nm, detuned by $\sim$ 800\,GHz from the $S_{1/2}\leftrightarrow P_{1/2}$ electric dipole transition of the ion. The Raman transitions are used for single qubit rotations $R(\theta,\phi)=\exp{[-\frac{\i}{2}\theta(\Xgate\cos{\phi}+\Ygate\sin{\phi})]}$: The amplitude $\theta$ is controlled by changes in the powers of the driving lasers and/or the pulse duration, while the phase $\phi$ is varied by changing the differential phase of the Raman beams. A third qubit rotation axis is accessed via a virtual $\Zgate$ gate, realized by adjusting the phase reference of subsequent control fields, rather than applying any physical laser pulse. This combined gate set has shown the randomized benchmarking fidelity of a single-qubit of $99.98(1)\,\%$~\cite{hilder2025nutshell,fsk2025variational}. Spin-selective measurement is carried out in two steps: First, shelving the $\ket{0}$ population in the metastable $D_{5/2}$ state. Then, observing the laser induced fluorescence when illuminating the ion with a laser beam near 397\,nm resonant with the S$_{1/2}$ to P$_{1/2}$ transition. The scattered photons are counted over a detection window $\tau$, and a threshold-based discrimination is carried out to distinguish between the non-fluorescent $D_{5/2}$ (dark)  and the $S_{1/2}$ ground (bright) states \cite{Roos2000phd}. Varying the detection window $\tau$, we can optimize for the \ac{SPAM} error. Deviations from the optimal window $\tau\approx1$\,ms degrade discrimination in one of two ways: short $\tau$ result in a low mean average number of counted photons such the Poisson statistics is not sufficient to safely distinguish the qubit states, while a too long $\tau$ increases the chances of a decay from the  1.2~s long lived $D_{5/2}$ state decay back to the ground state, see \cref{methods:fig:experiment_figures} (c) in \cref{methods:experiment}). For the actual setup, we found errors in the range $0.1-1\%$, minimized at $\tau$ = 1~ms, see \cref{methods:experiment}. Moreover, we conclude the detection cycle by measurements of ion losses, such that those events are post-selectively removed from the data sample.

\begin{figure}[!t]
    \centering
    \begin{tikzpicture}[panel/.style={anchor=north west,inner sep=0pt,outer sep=0pt,align=center},sublabel/.style={anchor=north west,font=\normalsize}]
    \node[panel,minimum width=0.95\linewidth] (top) at (0,0) {\includegraphics[width=0.95\linewidth]{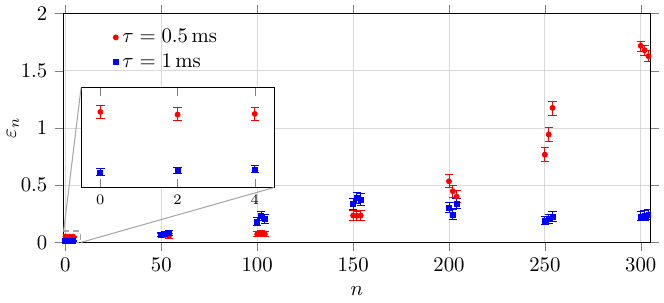}};
    
    \node[panel,minimum width=0.95\linewidth] (bottom) at ([yshift=-4mm,xshift=2mm]top.south west) {\includegraphics[width=0.95\linewidth]{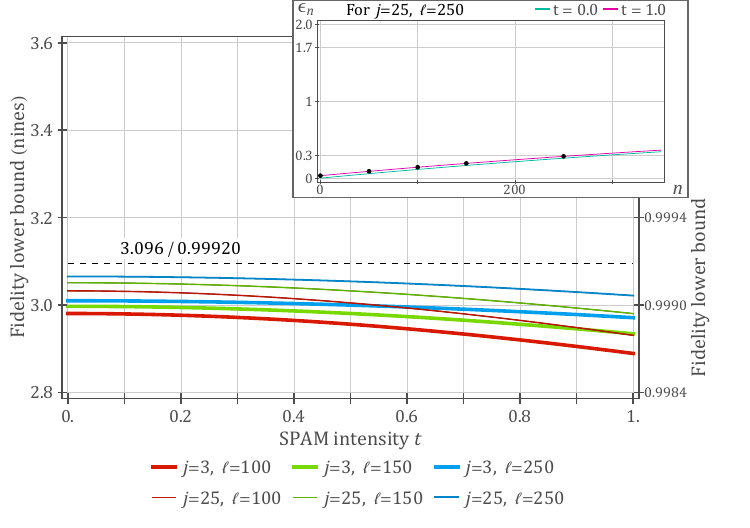}};
    
    \node[sublabel,xshift=-6pt,yshift=0pt] at (top.north west) {(a)};

    \node[sublabel,xshift=-12pt,yshift=0pt] at (bottom.north west) {(b)};
    \end{tikzpicture}

    \caption{\textbf{Experimental and numerical results.} (a) Experimental estimates of $\veps_n$ in \cref{eq:epsilon_n} as a function of sequence length $n$, obtained with fluorescence-detection windows $\tau=1.0$\,ms (blue squares) and $\tau=0.5$\,ms (red circles). 
    The inset resolves the short sequences $n\in{0,2,4}$, for which the shorter, suboptimal detection window produces a larger offset.  The anomalous behavior for $1.0$ ms is attributed to fluctuations in the effective Rabi frequency and a differential AC Stark shift between the Zeeman sublevels, mediated by their coupling to the $P_{1/2}$ state. The two datasets were acquired three hours apart.
    (b)  Certified lower bounds on the average gate fidelity for the fitted numerical noise model, expressed as $\mathrm{Nines}(\Fid_\avg)=-\log_{10}(1-\Fid_\avg)$. The noisy gate is held fixed while $t\in[0,1]$ scales the \ac{SPAM} errors, with $t=0$ corresponding to ideal \ac{SPAM} but a nonideal gate. Curves show different choices of $j$ and $\ell$; the dashed black line marks the exact gate fidelity, $\Fid_\avg=0.99920$, corresponding to $3.096$ nines. The inset compares $\veps_n$ at $t=0$ (cyan) and $t=1$ (magenta) for $j=25$ and $\ell=250$. Black points mark the sequence lengths used to construct the certificate in (a).}
    \label{fig:signal}
\end{figure}

To apply the certification method, we developed a Python library, available via \cite{QSQlab} (see \cref{methods:software} for details).
The measured signal, expressed in terms of estimated quantities $\veps_n$, for two detection windows $1$ and $0.5$ ms, is shown in \cref{fig:signal} (a).
The resulting quality metrics are summarized in \cref{tab:qsq_results}.
We use the sample mean to estimate the probabilities $\Pr[0\vert n]$ in \cref{eq:probs} for each sequence of length $n$ of $R(\frac{\pi}{2},0)$ gates (which are equal to $\sqrt{\Xgate}$ in \cref{eq:sqrtX} up to a global phase).
From these estimates, we obtain the values of $\veps_n$ according to \cref{eq:epsilon_n}.

The sequences of lengths $n \in \Set{0,2,4}$ were sampled $3\times10^4$ times each, from which we obtain $\sqrt{\veps_0}+\sqrt{\veps_2}\leq 0.21(3)$ for $\tau=1$ ms and $\sqrt{\veps_0}+\sqrt{\veps_2}\leq 0.46(3)$ for $\tau=0.5$ ms, each with $99\%$ confidence (see \cref{methods:statistics} for details).
Additionally, we obtain the short-sequence lower bounds on the average gate fidelity via \cref{eq:fid_short} (see \cref{tab:qsq_results}).
The higher-depth circuits were sampled $2\times 10^3$ times each.
To produce the lower bound on $\abs{\lambda}$ in \cref{eq:amplitude_bound}, we pick $\ell$ equal to $104$ and $250$ for $\tau=0.5$ and $\tau=1.0$ ms, respectively. 
For the phase certificate, we use $j=25$ in \cref{eq:bound_phi} for both detection times.
Subsequently, from \cref{eq:fid}, we obtain a sound certificate for the average gate fidelity (see \cref{tab:qsq_results}).
In \cref{fig:signal}~(a), we also display estimates of $\veps_n$ for other $n$, which are used to get a good guess of the parameters $\ell$ and $j$ for our certification method.
The details on statistical analysis can be found in \cref{methods:statistics}, and experimental data are available at~\cite{experiment}.

The signals in \cref{fig:signal}~(a) show how  \ac{SPAM} errors can affect the observed statistics. At short sequence lengths, highlighted in the inset, reducing the detection window from $\tau=1.0$ to $0.5$\,ms produces a shift in $\veps_n$,  reflecting the increased readout error. At larger $n$, repeated gate applications amplify the underlying gate dynamics, and the two datasets exhibit different behavior. Because the measurements were performed three hours apart, these differences may also reflect slow variations in the effective Rabi frequency and differential AC Stark shift.

The resulting certificates reveal a clear distinction between the two stages of the protocol. The increased readout error lowers the short-sequence fidelity bound from $98.5(4)\%$ to $92.7(8)\%$. By contrast, the long-sequence analysis produces consistent constraints on the intrinsic rotation dynamics: the rotation-angle deviation is bounded by $0.015(4)$ and $0.009(3)$ radians, while the lower bounds on $\abs{\lambda}$ remain within $10^{-3}$ of unity. Consequently, both detection settings yield the same average gate-fidelity lower bound of $99.94(3)\%$ at $99\%$ confidence. Thus, although the additional readout noise strongly changes the raw signal and the short-sequence estimate, it does not obscure the high-quality gate dynamics extracted by the long-sequence certificate. Random subsampling further shows that these long-sequence results can be reproduced using only $22\,000$ circuit executions.

\begin{table}[t]
\centering
\caption{\textbf{Experimental results.} l.b. and u.b., stand for lower and upper bounds, respectively. The statistical uncertainty corresponds to $1-\delta=0.99$ confidence for each of the reported quality metrics.}
\label{tab:qsq_results}
\begin{tabular}{cccccc}
\toprule
Detec. & Detec. & Short-seq. & $\abs{\lambda} $ & $\abs{\theta-\frac{\pi}{2}}$ & Long-seq. \\
window & error & $\Fid_{\avg}$ l.b. & l.b. & u.b. & $\Fid_{\avg}$ l.b.\\
$\tau$ [ms] & [\%] & \eqref{eq:fid_short} [\%] & \eqref{eq:amplitude_bound}  & \eqref{eq:bound_phi} & \eqref{eq:fid} [\%]\\
\midrule
1.0 & 0.151(7) & 98.5(4) & \phantom{.}0.9991(2)\phantom{.} & \phantom{.}0.015(4)\phantom{.} & 99.94(3) \\
0.5 & 1.25(6) & 92.7(8) & \phantom{.}0.9992(4)\phantom{.} & \phantom{.}0.009(3)\phantom{.} & 99.94(3) \\
\bottomrule
\end{tabular}
\end{table}

\subsection*{Numerical experiments}
The numerical analysis in \cref{fig:signal}~(b) isolates the influence of \ac{SPAM} errors under controlled conditions. We use the noise models discussed in~\cref{methods:simulations}, where parameters are obtained by least-squares fitting to the experimental data for $\tau=1$~ms from~\cref{fig:signal}~(a) (blue)~\cite{experiment}. The noisy gate is then held fixed, with exact average gate fidelity $\Fid_\avg=0.99920$, while the parameter $t$ continuously increases the \ac{SPAM} errors from zero to their full modeled strength. This idealistic controlled experiment lets us assess the \ac{SPAM} robustness of the certificate.

All certified bounds remain below the exact fidelity, as required by soundness. Increasing the \ac{SPAM} strength produces only a gradual deterioration of the bounds, and the tightest certificate, obtained for $j=25$ and $\ell=250$, remains above three nines of fidelity over the full range of $t$. The inset shows that the selected values of $\veps_n$ change only weakly between $t=0$ and $t=1$, explaining the stability of the resulting certificate. The comparison between different parameter choices also shows that larger values of $\ell$ generally improve the bound, while changing $j$ produces a more modest effect in this example. The protocol therefore does not require finely optimized sequence lengths: a coarse preliminary scan is sufficient to identify parameters that yield a near-optimal certificate.

Further simulations covering dephasing, coherent over-rotation, imperfect state preparation, asymmetric readout, and additional noise mechanisms are presented in \cref{methods:simulations}. The corresponding Jupyter notebooks are publicly available in~\cite{experiment}. The open-source QSQlab package~\cite{QSQlab} provides the complete certification workflow, including Qiskit-based circuit generation, signal analysis, and statistical post-processing; see \cref{methods:software}.

\section*{Discussion}
This work demonstrates that sound quantum-gate certification can be efficient and experimentally practical in a black-box qubit setting. Under the general assumption of time-independent qubit operations, our protocol requires no prior characterization, trusted reference operations, or randomized circuits. From a small set of experimentally tested circuits, it derives finite-sample bounds on the gauge-invariant rotation eigenvalue, which imply a fidelity guarantee relative to the target up to a change of basis. Crucially, this guarantee holds for every physical qubit model compatible with the data and does not rely on selecting a favorable representative through a more general, potentially nonphysical gauge transformation. Using fixed, deterministic sequences, 
we constrain the rotation dynamics $\pi/2$-pulse gate and certify its average fidelty on a quantum processor hardware. 

For the future, our results pave a route towards sound certification of increasingly complex quantum operations. One may extend to the complete single-qubit gate sets, including Clifford gates, and to multi-qubit entangling operations, where short-sequence protocols already exist~\cite{noller2025classical,noller2025sound, schroeder2025certifying}. In the multi-qubit setting, a particularly important objective is the  certification of crosstalk, separating intrinsic gate errors from unwanted interactions with neighboring qubits. Moreover, the small number of deterministic circuits and lightweight analytic post-processing make the protocol well suited for integration into in situ recalibration and autonomous quantum-control loops, where its certificates can serve as robust optimization objectives.

\begin{acknowledgements}
We thank Zaw Lin Htoo, Martin Kliesch, Markus Heinrich, and Valerio Scarani for useful discussions. We acknowledge funding by the Deutsche Forschungsgemeinschaft (DFG, German Research Foundation) – project number 563372006 – under Germany's Priority Program SPP 2514 ``Quantum Software, Algorithms and Systems – Concepts, Methods, and Tools for the Quantum Software Stack''. NM, JN, LV, MG acknowledge (ZAQC) under funding by the Hessian Ministry of Digital Strategy and Innovation and the Hessian Ministry of Higher Education, Research and the Arts the National Research Center for Applied Cybersecurity ATHENE, within the project ``Zentrum f\"ur Angewandtes Quantencomputing'', and Fujitsu Germany GmbH as part of the endowed professorship ``Quantum Inspired and Quantum Optimization''. JM, UP and FSK acknowledge funding by the German Federal Ministry of Research, Technology and Space (BMFTR) within the projects IQuAn, ATIQ.
\end{acknowledgements}

\section*{Methods}
\subsection{Quantum models and deterministic certification tests}\label[methods]{methods:QSQ}

We first specify the object being certified. Since we do not assume trusted state preparation, gates, or measurements, the implementation is described as a full state-gate-measurement model rather than as a gate alone.
\begin{definition}\label{def:model}
     A qubit quantum model is a triple $(\rho,\Set{\Lambda},\{M_0,M_1\})$, consisting of an initial state $\rho$, a quantum channel $\Lambda$, and a binary \ac{POVM} $\{M_0,M_1\}$, all defined on a Hilbert space $\H\cong\CC^2$.
\end{definition}

We state the definition for the single-qubit, single-gate setting considered in this work. The same framework extends to finite-dimensional systems, measurements with more outcomes, and models containing several channels corresponding to different implemented gates, including universal gate-sets (see Refs.~\cite{noller2025classical,noller2025sound,schroeder2025certifying}). Our target model is
 $(\kb{0}{0},\Set{\sqrt{\Xgate}},\Set{\kb{0}{0},\kb{1}{1}})$, that is, preparation of $\ket{0}$, implementation of a $\pi/2$ rotation about the $\Xgate$ axis, and measurement in the computational basis. Whenever there is no ambiguity, we use the same notation $\sqrt{\Xgate}$ for the unitary and for the corresponding unitary channel.

The quality of an implemented model is defined up to a common change of the basis in $\H$. Thus, we compare the implemented state, channel, and measurement to the target model after conjugating the target by a single unitary $U\in\U(2)$. For the state, we use the fidelity
\begin{equation}
\Fid(U\ket{0},\rho)
=
\Tr[\rho\,U\kb{0}{0}U^\dagger].
\end{equation}
For the channel, we use the average gate fidelity with respect to the target unitary $\sqrt{\Xgate}_{\vert U}\coloneqq U\sqrt{\Xgate}U^\dagger$,
\begin{equation}\label{eq:fid_avg_def} \Fid_\avg(\sqrt{\Xgate}_{\vert U},\Lambda) = \int\mathrm{d}\psi\, \Tr\left[\sqrt{\Xgate}_{\vert U}\kb{\psi}{\psi}\sqrt{\Xgate}_{\vert U}^{\dagger} \,\Lambda(\kb{\psi}{\psi}) \right], 
\end{equation}
where the integral is taken with respect to the Haar-invariant probability measure on pure states. For the two-outcome \ac{POVM}, we use the total-variation distance in the same unitary gauge $U$, which reduces to
\begin{equation}
\norm{U\kb{0}{0}U^\dagger-M_0}_\infty .
\end{equation}

In what follows, we work in the \ac{PTM} representation for qubit channels. 
The \ac{PTM} of a qubit channel has the form
\begin{equation}\label{eq:PTM_Lambda}
    \begin{pmatrix}
        1 & \begin{matrix}0&0&0\end{matrix}\\
        \vec{l} & L
    \end{pmatrix},
\end{equation}
where $\vec l\in\RR^3$ and $L\in\RR^{3\times 3}$ corresponds to the traceless part of the implemented channel.
For the target $\pi/2$ pulse gate, this matrix has the form
\begin{equation}\label{eq:L0}
    L_0\coloneqq
    \begin{pmatrix}
        1&0&0\\
        0&0&-1\\
        0&1&0
    \end{pmatrix},
\end{equation}
see \cref{app:PTM} for details.
A unitary transformation $U$ on $\CC^2$ corresponds to an orthogonal transformation $V$ on $\RR^3$, with the traceless part of the channel's \ac{PTM} transforming as
\begin{equation}
    L\mapsto VLV^\T .
\end{equation}
In this notation, the average gate fidelity in the gauge $U$ can be written as
\begin{equation}\label{eq:fid_avg_ptm_methods}
    \Fid_\avg(\sqrt{\Xgate}_{\vert U},\Lambda)=\frac{1}{2}+\frac{1}{6}\Tr\left[L_0^\T VLV^\T\right].
\end{equation}

One of the central characteristics of this work is that the tests used in this work are deterministic-sequence tests. They are based on a single-gate qubit instance of the \ac{QSQ} framework recently introduced by some of us in Refs.~\cite{noller2025classical,noller2025sound,schroeder2025certifying}. 
The essential point of the \ac{QSQ} framework is that for selected gate sequences, the ideal model predicts deterministic outcomes, and importantly, we can reliably certify the ideal implementation by testing this behavior on a finite set of such input sequences. 
For the target $\pi/2$ pulse, after $n$ applications of the gate to the state $\ket{0}$, the probability of outcome $0$ is equal to $1$ for $n\in\Set{0,4,8,\dots}$ and equal to $0$ for $n\in\Set{2,6,10,\dots}$. Equivalently, for even $n$, the ideal outcome $a$ is determined by $a\equiv n/2 \pmod 2$. Experiments with odd $n$ are not tested in our framework. Thus, each run of the experiment asks a deterministic question: after applying the same channel $n$ times and measuring at the end, did the observed outcome agree with the one predicted by the target model? 

The choice of the tested sequences depends on the certification task. Different target gates, error regimes, and desired quality guarantees may require different questions. In earlier works by some of us, such sequences were constructed for single- and multi-qubit models, including universal gate sets~\cite{noller2025classical,noller2025sound}. More generally, the problem of finding deterministic sequences with good certification power is closely related to questions in classical finite-automata theory, as investigated in Ref.~\cite{schroeder2025certifying, ambainis2012superiority,ambainis2021automata,kondacs1997power}. Thus, the construction of useful test sequences is not merely a technical detail of the protocol, but a research question in itself.

In this work, we construct deterministic sequences tailored to the high-fidelity certification of the $\pi/2$ pulse.
The basic protocol that we use estimates the number of wrong outcomes for a given number $n$ of repeated applications of $\pi/2$ rotation.

\begin{algorithm}[H]
\caption{Counting wrong outcomes}
\label{protocol}
\begin{algorithmic}
\STATE \textbf{Input:} Sequence length $n$, number of repetitions $N$.
\STATE Set $e_{n}\leftarrow 0$.
\FOR{$i=1,\dots,N$}
\STATE Run the circuit in \cref{fig:scheme} (b), record the outcome $a$.
    \IF{$a\not\equiv n/2 \pmod 2$}
        \STATE $e_{n}\leftarrow e_{n}+1$.
    \ENDIF
\ENDFOR
\STATE \textbf{return} $e_{n}$.
\end{algorithmic}
\end{algorithm}

It is clear that for i.i.d.~implementations, ${e}_n/N$ is an unbiased estimator for $\Pr[a\vert n]$, for $a$ being the wrong outcome for length $n$.
As pointed out in the main text, we report the certificates for gate implementation in terms of $\veps_n$ defined in \cref{eq:epsilon_n}, which are combined error probabilities in two tests for subsequent lengths $n$ and $n+2$.
As the unbiased estimator for $\veps_n$, we take $\hat{\veps}_n = (e_n+e_{n+2})/N$, i.e., we sum up two outputs of \cref{protocol} for $n$ and $n+2$.

In what follows, we give more details about the two types of tests we present in this work: the short sequences are used to verify that the implementation lies in the high-fidelity regime, while longer sequences are used to extract information about the gauge-invariant rotation eigenvalue $\lambda$.
The latter is certified via an upper bound on the deviation of its phase from the target $\pi/2$ value and a lower bound on its absolute value $\abs{\lambda}$.
\subsection{Short-sequence certification of the high-fidelity regime}
\label[methods]{methods:short}

The short-sequence test formalizes the preliminary certification step described in the main text. It uses only the deterministic sequence lengths $n\in \Set{0,2,4}$, corresponding to the return-flip-return pattern of the ideal $\pi/2$ pulse. Its purpose is not to estimate the rotation eigenvalue $\lambda$ with high precision, but to certify that the observed data implies a high-fidelity state-gate-measurement model. This is the short-sequence soundness guarantee. 

\begin{theorem}[Soundness of the short-sequence test]
\label{th:sound_short}
Let $\veps_0$ and $\veps_2$ be defined as in \cref{eq:epsilon_n} and assume that
$ \sqrt{\veps_0}+\sqrt{\veps_2}<1$.
Then there exists a unitary $U\in\U(2)$ such that
\begin{equation}\label{eq:th_short_gate_fid}
    \Fid_\avg(\sqrt{\Xgate}_{\vert U},\Lambda) \geq \frac{1}{3} + \frac{2}{3}\sqrt{1-(\sqrt{\veps_0}+\sqrt{\veps_2})^2}.
\end{equation}
Moreover, for the same unitary $U$, $\Fid(U\ket{0},\rho)\geq 1-\veps_0$, and $\norm{U\kb{0}{0}U^\dagger-M_0}_\infty\leq \veps_0$.
\end{theorem}

The proof is given in~\cref{app:sound_short} with the accompanying completeness results in~\cref{app:complete_short}. The theorem shows that the short sequences certify the state preparation, channel, and measurement jointly. In particular, the test is self-consistent: it does not assume that the observed deterministic pattern is generated using trusted \ac{SPAM} components. The corresponding completeness statement shows that the test is robust: models that are close to the target model pass with high probability.

Before turning to the long-sequence estimates, we present the consequence of the short-sequence certificate on \ac{SPAM} and gate noise separation. 

\begin{lemma}~\label{lemma:eigenvalues}
   Let $\veps_0$, $\veps_2$, as defined in \cref{eq:epsilon_n}, satisfy $\sqrt{\veps_0}+\sqrt{\veps_2}<\frac{1}{2}$.
    Then, the block $L$ of the~\ac{PTM} of the implemented channel $\Lambda$ is diagonalizable, with eigenvalues $\Set{\lambda_0,\lambda_1,\lambda_1^\ast}$, where $\lambda_0\in\RR$, and $\lambda_1\in\CC\setminus\RR$.  
\end{lemma}
See \cref{app:sec_gate_SPAM} for a proof.
Once we certify with the short sequence test that $\sqrt{\veps_0}+\sqrt{\veps_2}<\frac{1}{2}$, using \cref{lemma:eigenvalues} we conclude that the block $L$ of the \ac{PTM} has the correct structure of an approximate rotation.
Next, we obtain a lower bound on the fidelity, which is expressed in terms of the eigenvalues of $L$.

\begin{lemma}\label{lemma:fidelity}
    Let $\veps_0$, $\veps_2$, as defined in \cref{eq:epsilon_n}, satisfy $\sqrt{\veps_0}+\sqrt{\veps_2}<\frac{1}{2}$. Then there exists $U\in \U(2)$ for which \begin{equation}\label{eq:methods_fid_lambdas}
    \Fid_\avg(\sqrt{\Xgate}_{\vert U},\Lambda)
\geq \frac12+\frac16\bigl(\lambda_0+2\Im[\lambda_1]\bigr),
\end{equation}
where $\Set{\lambda_0,\lambda_1,\lambda_1^\ast}$ are the eigenvalues of the block $L$ of the \ac{PTM} of the implemented channel $\Lambda$.
\end{lemma}

A proof of~\cref{lemma:fidelity} is straightforward and is deferred to~\cref{app:sec_gate_SPAM}. Here instead, we discuss its consequences. In the regime where the short sequence test is passed, one can find a unitary gauge such that the quality of the gate can be reported in terms of its eigenvalues, which are independent of the~\ac{SPAM}. The rest of the task of the paper is to derive certificates for the eigenvalues of $L$ from experimentally measurable data. 
To simplify this task, we use complete-positivity constraint $\lambda_0\geq 2\abs{\lambda_1}-1$ for the qubit channel $\Lambda$ (see \cref{app:PTM}) and transform~\cref{eq:methods_fid_lambdas} to
\begin{equation}\label{eq:fid_expr_lambdas_methods}
\begin{split}
    \Fid_\avg(\sqrt{\Xgate}_{\vert U},\Lambda)\geq \frac{1}{3}+\frac{1}{3}\left(\abs{\lambda_1}+\Im[\lambda_1]\right).
\end{split}
\end{equation}
In the main text and in what follows, we write $\lambda\coloneqq\lambda_1$ for the complex eigenvalue, to which we refer to as the rotation eigenvalue of $\Lambda$. In what follows, we describe a qubit black-box protocol to certify $\lambda$.

\subsection{Eigenvalue certification from long-sequence tests} \label[methods]{methods:long}
From here on, we assume that the eigenvalues of the matrix $L$ in the parametrization of the implemented channel $\Lambda$  are $\Set{\lambda_0,\lambda_1,\lambda_1^\ast}$, where $\lambda_0\in\RR$ and $\lambda\coloneqq\lambda_1\in \CC\setminus \RR$.  
The following result states more formally our approach to certify the eigenvalue $\lambda$, in particular its phase.

\begin{lemma}\label{lemma:xi_bound}
  Let $j$ be an odd integer, and let $\veps_{0}^\uparrow\leq \veps_{2j}^\uparrow\leq \veps_{4j}^\uparrow\leq \veps_{6j}^\uparrow$ be a reordering of $\Set{\veps_{0},\veps_{2j},\veps_{4j},\veps_{6j}}$, defined in \cref{eq:epsilon_n}.
    Assume $\veps_{6j}^\uparrow<1$, and define
    \begin{equation}\label{eq:xi_definition}
        \xi\coloneqq \sqrt{\frac{\veps^\uparrow_{6j} + \veps_{4j}^
\uparrow-\veps_{2j}^\uparrow-\veps_{0}^\uparrow}{1-\veps_{2j}^\uparrow}}\ .
    \end{equation}
Then, if $\xi< 1$, the complex eigenvalue $\lambda=\e^{\ii\theta}\abs{\lambda}$ of the \ac{PTM} of the implemented channel $\Lambda$ satisfies $\abs{1+\lambda^{2j}}\leq \xi$, which implies
\begin{align}\label{eq:amplitude_angle_bound}
       & \abs{\lambda}\geq (1-\xi)^{\frac{1}{2j}},\\ 
        & \theta+ k\frac{\pi}{j}\in\left[\frac{\pi}{2}-\frac{\arcsin(\xi)}{2j},\frac{\pi}{2}+\frac{\arcsin(\xi)}{2j}\right],\; 
        \label{eq:amplitude_angle_bound_1}
    \end{align}
for some $k\in\ZZ$.    
Moreover, if there exists $\eta\in [0,\frac{\pi}{j}]$, such that $\theta\in\left[\frac{\pi}{2}-\eta,\frac{\pi}{2}+\eta\right]$ and $\arcsin(\xi)<2\pi -2j\eta$, then we can conclude that \cref{eq:amplitude_angle_bound_1} holds for $k=0$. 
\end{lemma}
We present a proof sketch below, with a complete proof given in~\cref{app:long}.

\begin{proofsketch}{\cref{lemma:xi_bound}}
The starting point is that each $1-\veps_n$ is a linear function of $L^{n}$, 
\begin{equation}\label{eq:veps_T_proof_sketch}
    1-\veps_n = (-1)^{\frac n2} \vec m^{\T}L^{n}\vec q,
\end{equation}
where $\vec{m}$ and $\vec{q}$ depend on the \ac{SPAM}.
If we set $T\coloneqq L^{2j}$ and take $j$ odd, then $\veps_n$ for $n\in \Set{0,2j,4j,6j}$ correspond to the four consecutive powers $\Set{\mathbbm{1},T,T^2,T^3}$.
Since $T$ is a $3\times3$ matrix, these four powers must satisfy a linear relation by the Cayley--Hamilton theorem~\cite{Frobenius_1878}. 
Due to \cref{eq:veps_T_proof_sketch}, this relation must also be satisfied by $(-1)^\frac n2(1-\veps_n)$.
The coefficients of this linear relation depend on the eigenvalues of $T$, namely $\lambda_0^{2j}$, $\lambda^{2j}$, and $(\lambda^\ast)^{2j}$.
Resolving this relation with respect to the eigenvalues and using the physical constraints on the channel, we obtain the bound
\begin{equation}
 \abs{1+\lambda^{2j}}\leq \xi,    
\end{equation}
with $\xi$ defined in the statement of the lemma.
Geometrically, this means that $\lambda^{2j}$ lies close to $-1$ in the complex plane. Hence $\lambda$ has modulus close to one and phase close to $\pi/2$, up to shifts by multiples of $\pi/j$. If prior information already localizes the phase near $\pi/2$, this removes the branch ambiguity and gives the bound with $k=0$.
\end{proofsketch}

For $j=1$, whenever $\xi<1$, i.e., as long as the bounds in \cref{eq:amplitude_angle_bound,eq:amplitude_angle_bound_1} are non-trivial, we can also infer that $k=0$ in the bound for the phase, since we assume that $\Im[\lambda]>0$ (otherwise, we prove the results for $\lambda^\ast$).
For $j>1$, with \cref{lemma:xi_bound}, we essentially test the implementation of $j$ consecutive $\pi/2$-rotations as a single operation, which leads to the fact that we can only certify the phase $\theta$ up to an angle $k\frac{\pi}{j}$. 
The way around this is to apply \cref{lemma:xi_bound} for two different $j$, e.g., first for $j=1$, resulting in a certificate $\theta\in [\frac{\pi}{2}-\eta,\frac{\pi}{2}+\eta]$, and then for a larger $j$ that satisfies the conditions of the lemma $j\leq \frac{\pi}{\eta}$ and for which $\arcsin(\xi)<2\pi -2j\eta$ holds.
We also note that one can shift the set of tested sequence lengths in \cref{lemma:xi_bound} by some even integer $i$, and have the same result as stated by the lemma for the quantities $\Set{\veps_i,\veps_{2j+i},\veps_{4j+i},\veps_{6j+i}}$ (see \cref{app:long} for details). 

The lower bound on $\abs{\lambda}$ in \Cref{lemma:xi_bound} can be improved. Intuitively, the four test points $\Set{\veps_0,\veps_{2j},\veps_{4j},\veps_{6j}}$ do not contain enough information to cleanly separate the effects of coherent over-rotation from those of decoherence. We therefore introduce an additional coherence test at a new sequence length $\ell$. This extra test point allows us to isolate the decay of the complex eigenvalue more directly and derive a stronger lower bound on $\abs{\lambda}$.

\begin{lemma}\label{lemma:amplitude_bound}
    Assume that the complex eigenvalue $\lambda=\abs{\lambda}\e^{\i\theta}$ of the \ac{PTM} of the implemented channel satisfies $\abs{\lambda}\geq 1-\eta$ for $\eta\leq \frac{1}{10}$ and $\abs{\theta-\frac{\pi}{2}}\leq \frac{\pi}{10}$.
    Then, if $\abs{1-\veps_\ell}\geq 7\eta$ for an even $\ell$, where $\veps_\ell$ is defined in \cref{eq:epsilon_n}, we have
    \begin{equation}\label{eq:abs_lambda_lb}
        \abs{\lambda}^{\ell}\geq \frac{\abs{1-\veps_\ell}-2(1-\abs{\lambda})}{\sqrt{\frac{2}{1-8\abs{\lambda}(1-\abs{\lambda})}-1}}\ .
    \end{equation}  
    Moreover, if $\abs{1-\veps_\ell}\geq 7/10$, and $\ell\geq 22$, then $\abs{\lambda}^{\ell-11}\geq \abs{1-\veps_\ell}$.
\end{lemma}
We give a proof sketch below, and refer to \cref{app:long} for a detailed proof.

\begin{proofsketch}{\cref{lemma:amplitude_bound}}
    The starting point is again the expression for $\veps_\ell$ in \cref{eq:veps_T_proof_sketch}. This time, however, we use the form of $\vec{q} = \frac{1}{2}(\1-L^2)\vec{r}-\frac{1}{2}(\1+L)\vec{l}$, where $\vec{r}$ is the Bloch vector of the initial state and $\vec{l}\in \RR^3$ is the non-unital part of the \ac{PTM} of the implemented channel $\Lambda$ (see \cref{eq:PTM_Lambda}).
    Unital channels satisfy $\vec{l}=0$, and in our case, under the conditions of the lemma, we can prove that $\norm{\vec{l}}\leq 2(1-\abs{\lambda})$. 
    Using the standard norm inequalities, we can then conclude that 
    \begin{equation}\label{eq:lemma_lambda_proofsketch_1}
        \abs{1-\veps_\ell}\leq \frac{\kappa_A}{2}\max\Set{\lambda_0^\ell(1-\lambda_0^2),\abs{\lambda}^\ell\abs{1-\lambda^2}}+\norm{\vec{l}},
    \end{equation}
    where $\kappa_A$ is the condition number of the matrix $A$ that diagonalizes $L$. The assumptions on the spectrum of $L$ further allow us to upper-bound $\kappa_A$ solely in terms of $\abs{\lambda}$. It therefore remains to determine which of the two spectral contributions in the maximum in \cref{eq:lemma_lambda_proofsketch_1} can dominate. Using the physicality relation between $\lambda_0$ and $\abs{\lambda}$, together with the lower bound on $\abs{1-\veps_\ell}$ assumed in the lemma, one rules out the possibility that the maximum is attained by $\lambda_0$.
    Combining the corresponding term in \cref{eq:lemma_lambda_proofsketch_1} with the bounds on $\kappa_A$ and $\norm{\vec l}$ yields the lower bound in \cref{eq:abs_lambda_lb}. The simplified bound follows by substituting the assumed preliminary bound on $\abs{\lambda}$ and rearranging the resulting inequality.
\end{proofsketch}

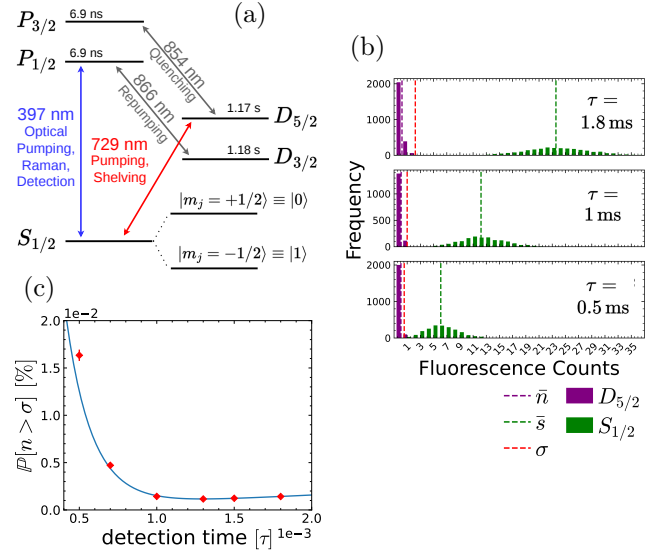
\begin{figure}[!t]
    \centering
    \begin{tikzpicture}[
        panel/.style={anchor=north west,inner sep=0pt,outer sep=0pt,align=center},
        sublabel/.style={anchor=north west,font=\normalsize}]
        \node[panel] (panel-a) at (0,0) {%
            \includegraphics[width=0.5\linewidth]{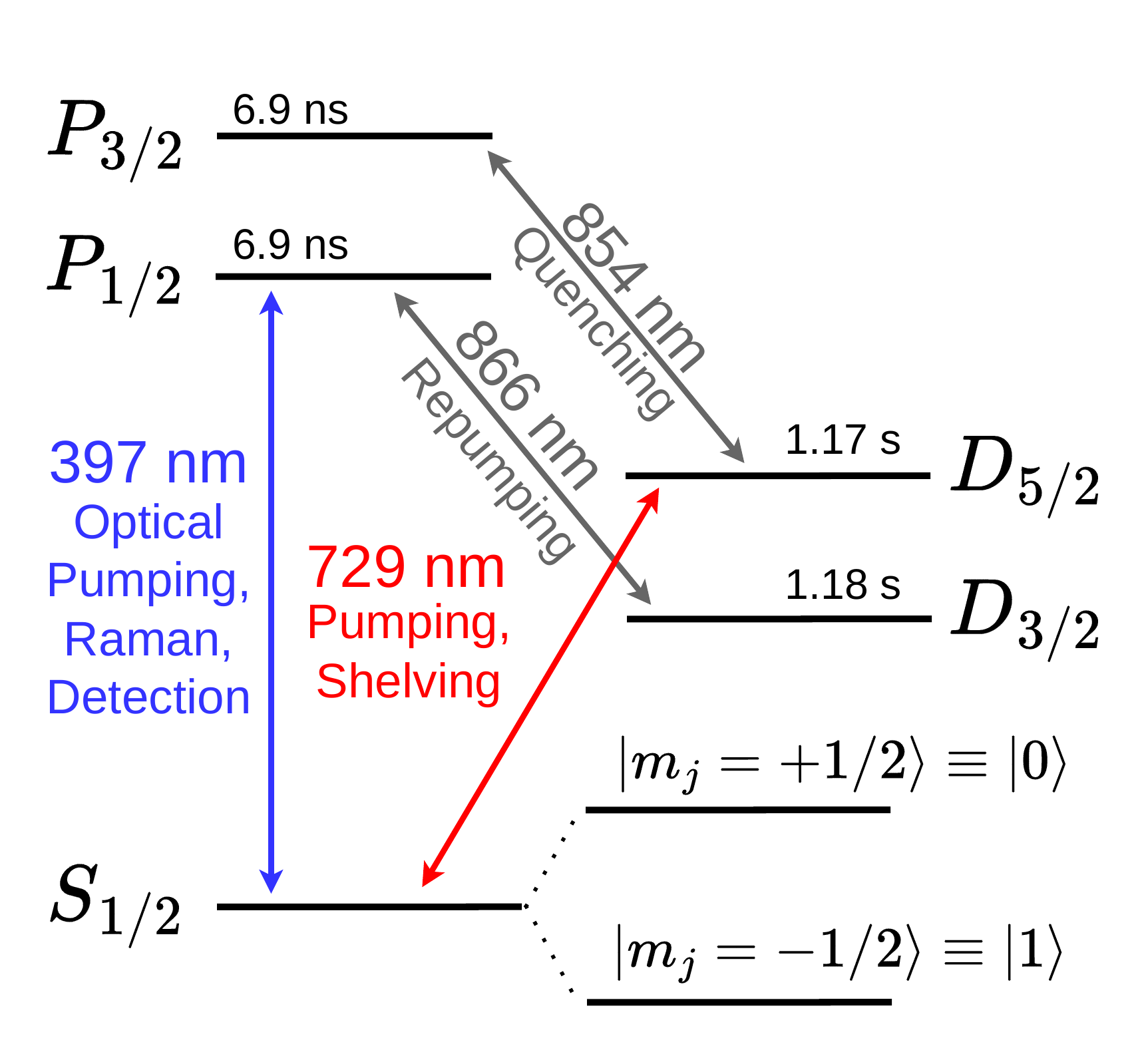}%
        };
        \node[panel] (panel-c)
            at ([yshift=-2mm]panel-a.south west) {%
            \includegraphics[width=0.5\linewidth]{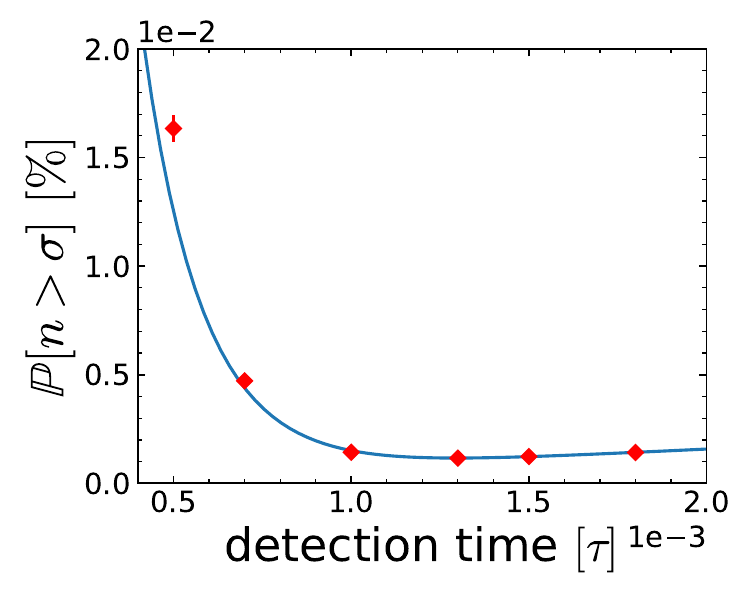}%
        };
        \node[panel] (panel-b)
            at ([xshift=0.5\linewidth,yshift=-10mm]panel-a.north west) {%
            \includegraphics[width=0.5\linewidth]{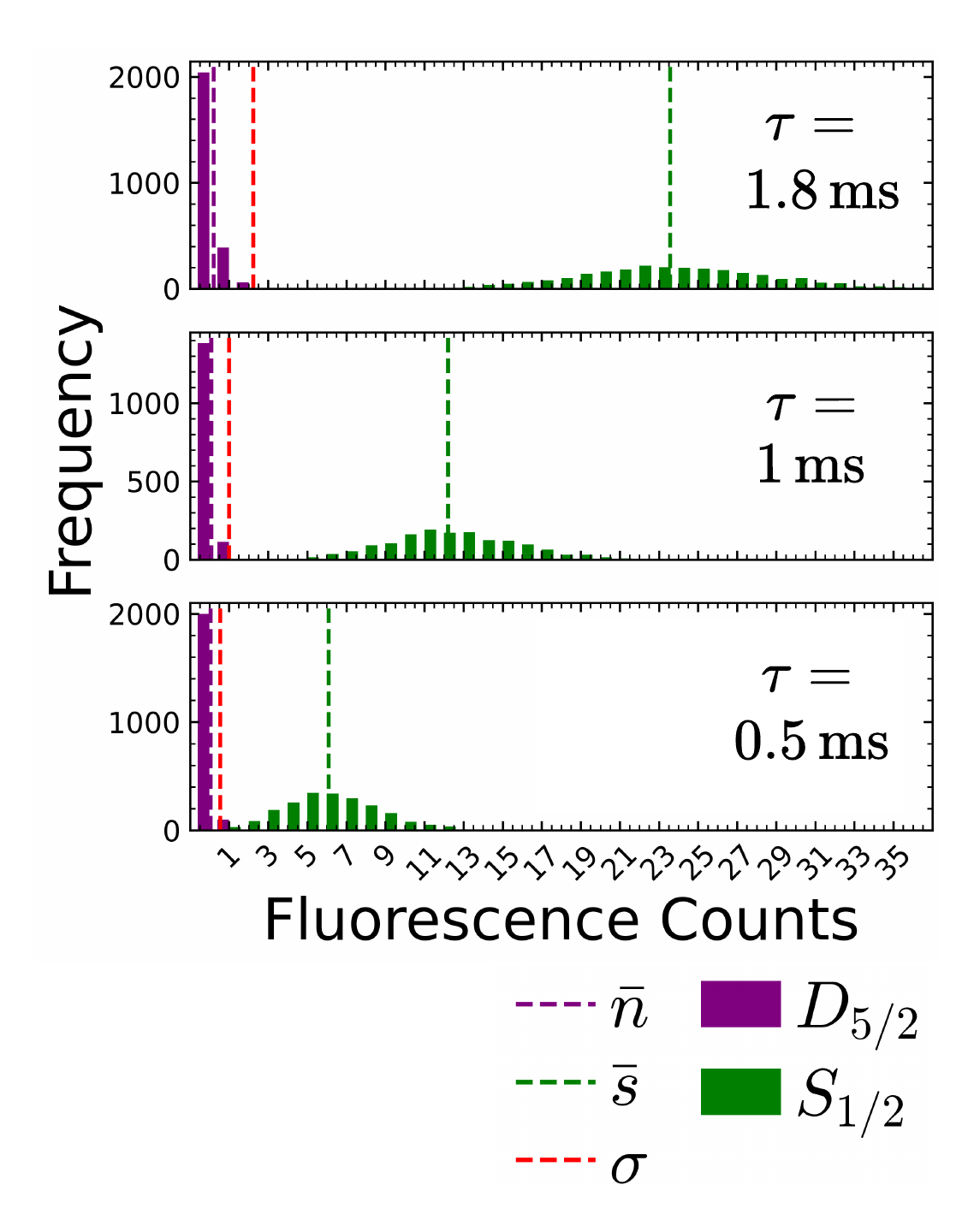}%
        };
        \node[sublabel, xshift=-40pt, yshift=-4pt]
            at (panel-a.north east) {(a)};
        \node[sublabel, xshift=4pt, yshift=12pt]
            at (panel-b.north west) {(b)};
        \node[sublabel, xshift=4pt, yshift=12pt]
            at (panel-c.north west) {(c)};
    \end{tikzpicture}
    \caption{\textbf{Fluorescence-based state discrimination for a single trapped $^{40}\mathrm{Ca}^+$ ion.} (a) Relevant level scheme of $^{40}\mathrm{Ca}^+$, showing the transitions used for optical pumping, Doppler cooling, state detection, and shelving to the metastable $D_{5/2}$ state. (b) Photon-count distributions recorded on the detector for the $S_{1/2}$ (bright), and $D_{5/2}$ (dark) states, shown for several detection window durations $\tau$, illustrating the growing separation of the two distributions with increasing $\tau$. (c) Probability $\Pr[n>\sigma]$ of the dark state photon count $n$ exceeding a discrimination threshold $\sigma=\sqrt{\bar s\bar n}$, plotted as a function of the detection window $\tau$. Individual data points are compared to the model prediction based on Poissonian photon statistics with mean counts $\bar s= R_s\tau$, $\bar n= R_n\tau$ for the bright and dark states, where $R_s$ and $R_n$ are the corresponding scattering (count) rates.}
    \label{methods:fig:experiment_figures}
\end{figure}

\subsection{Details about experimental implementation}
\label[methods]{methods:experiment}
This subsection outlines the derivation of fluorescence detection errors in the ion-trap experiments. In single-qubit sequences, \ac{SPAM} errors arise from several sources, including imperfect optical pumping during state preparation and residual off-resonant coupling of the 729\,nm shelving laser to unintended Zeeman sublevels. Of all sources, we focus on errors arising during the state-selective fluorescence detection window, since the photon-count statistics underlying this process admit a well-defined analytical model of the detection error probability, allowing it to be isolated and characterized independently of the other \ac{SPAM} contributions.

With this in mind, the certification protocol runs in the ion trap were accompanied by changes in the fluorescence detection window $\tau$ of the state-dependent discrimination between the $S_{1/2}$ and $D_{5/2}$ states (see Experimental implementation section in the main text). Deviations of $\tau$ from its optimal value of approximately $1$~ms introduce readout errors according to the model described in \cref{app:eq:readout_error} and Ref.~\cite{Roos2000phd}. This describes the probability that the scattered photons from the metastable $D_{5/2}$ state will surpass the threshold $\sigma$.
\begin{equation}\label{app:eq:readout_error}
    \Pr[n>\sigma]\approx\frac{\tau}{T}\frac{\bar s-\sigma}{\bar s-\bar n} + \frac{1}{2}\left(1-\mathrm{erf}\left(\frac{\bar s-\sigma}{\sqrt{2\bar s}}\right)\right)
\end{equation}
Here, $T$ represents the $D_{5/2}$ state lifetime, while $\bar s$ and $\bar n$ denote the mean photon counts of the $S_{1/2}$ and $D_{5/2}$ states, respectively, assuming infinite $D_{5/2}$ lifetime. The discrimination threshold is defined as $\sigma=\sqrt{\bar s\bar n}$, accounting for the Poissonian nature of the distributions.
The mean photon counts scale linearly according to $\bar s = R_s \tau$, $\bar n = R_n \tau$, where $R_s$ and $R_n$ are the scattering rates for the bright and dark states. \Cref{methods:fig:experiment_figures} (c) shows the behavior of this detection error model according to experimental data, yielding a signal-to-noise ratio of $R_s/R_n\approx130.8(58)$.
Additionally, the experimental data showed state-dependent detection errors. This results from cross-shelving effects caused by imperfect calibration of the 729\,nm laser driving the $\ket{0}\leftrightarrow D_{5/2}$ transition. 

At the time of these fluorescence measurements, the entire population was prepared in $\ket{0}\equiv\ket{4S_{1/2}, m_J=+1/2}$ and subsequently shelved into the metastable $D_{5/2}$ (dark) state, such that the ion population occupied only the $D_{5/2}$ manifold at the point of measurement. It is therefore sufficient to evaluate the single error probability $\Pr[n>\sigma]$, describing the likelihood that the photon counts scattered from this ($D_{5/2}$-shelved) population exceeded the discrimination threshold $\sigma$ and are misidentified as bright.

However, this simplification no longer holds for fluorescence measurements where the population is prepared in a statistical mixture of $\ket{0}$ and $\ket{1}$ -- for instance, by applying a $\sqrt{\Xgate}$ gate after initialization in $\ket{0}$. In this case, the total error probability must account for misidentification of both populations and is given by
\begin{equation}
    \Pr_{\text{error}}=\frac{N_{\ket{0}}}{N_{\text{total}}}\Pr[n>\sigma]+\frac{N_{\ket{1}}}{N_\text{total}}\Pr[s<\sigma],
\end{equation}
where $N_{\ket{0}}$, $N_{\ket{1}}$ denote the number of population instances prepared in $\ket{0}$ and $\ket{1}$, respectively, and $\Pr[s<\sigma]=\frac{1}{2}(1-\mathrm{erf}(\frac{\bar s-\sigma}{\sqrt{2\bar s}}))$.

\subsection{Statistical analysis}
\label[methods]{methods:statistics}
Here we describe the statistical procedure used to obtain confidence bounds on the fidelity lower bounds reported in \cref{tab:qsq_results}. We first construct exact one-sided confidence bounds on the wrong-outcome probabilities measured by \cref{protocol}, and then propagate them through \cref{lemma:xi_bound,lemma:amplitude_bound} and \cref{eq:fid_expr_lambdas_methods} to produce the lower bound on the average gate fidelity from the measured data.

For a given sequence length $n$, let
\begin{equation}
    p_n \coloneqq  \Pr\left[a\not\equiv n/2\pmod 2 \,\middle|\,n \right]
\end{equation}
be the wrong-outcome probability. Let $E_n$ denote the number of wrong outcomes in $N$ independent repetitions of \cref{protocol}. Then, under the i.i.d., assumption
\begin{equation}
    E_n\sim\Bin(N,p_n).
\end{equation}
For an observed count $e_n$, we estimate $\hat{p}_n=\frac{e_n}{N}$.
Since the quantities entering our analysis are $\veps_n=p_n+p_{n+2}$, we use $\hat{\veps}_n=\frac{e_n+e_{n+2}}{N}$ to estimate them, where we set the same number $N$ of repetitions for each tested sequence length.

We use one-sided exact Clopper--Pearson bounds~\cite{clopper1934use}. For an assigned unconfidence probability $\delta_n\in(0,1)$, define
\begin{equation}\begin{split}
    p_n^{\geq} &\coloneqq \inf\left\{q\,\middle|\,\Pr_{E_n\sim\Bin(N,q)}[E_n\geq e_n]\geq\delta_n\right\},\\
    p_n^{\leq} &\coloneqq\sup\left\{q\,\middle|\,\Pr_{E_n\sim\Bin(N,q)}[E_n\leq e_n]\geq\delta_n\right\},
\end{split}
\end{equation}
where clearly the domain for $q$ is $[0,1]$.
These satisfy
\begin{align}
    \Pr[p_n\geq p_n^{\geq}]&\geq 1-\delta_n, &\Pr[p_n\leq p_n^{\leq}] &\geq 1-\delta_n.
\end{align}
Corresponding bounds on $\veps_n$ are obtained by addition,
\begin{align}
    \veps_n^{\geq}&=p_n^{\geq}+p_{n+2}^{\geq},\\
    \veps_n^{\leq}&=p_n^{\leq}+p_{n+2}^{\leq},
\end{align}
with the confidence level $1-\delta_n-\delta_{n+2}$ following from the union bound. More generally, to obtain an overall confidence level of at least $1-\delta$, we assign unconfidence probabilities $\delta_1,\ldots,\delta_M$ to all elementary one-sided bounds used in the analysis such that
\begin{equation}
    \sum_{s=1}^{M}\delta_s\leq\delta.
\end{equation}
All subsequent bounds on $\veps_n$, the phase, and the absolute value of $\lambda$, and finally the average gate fidelity are then obtained by direct propagation.
A Jupyter notebook implementing the analysis for the datasets underlying \cref{tab:qsq_results}, along with the experimental datasets, is available online at \cite{experiment}.

For the short-sequence preliminary test that needs to ensure that $\sqrt{\veps_0}+\sqrt{\veps_2}< \frac{1}{2}$ with high probability, we can do a more nuanced analysis which is less pessimistic than the one relying on the union bound.
The simple statistical test we describe in the main text comprises running circuits of lengths $n\in\Set{0,2,4}$ $N$ times each and accepting if no wrong outcome occurs. If we denote the probabilities of observing a wrong outcome for these lengths as $\Set{p_0,p_2,p_4}$, and assume i.i.d., then the probability of acceptance is $(1-p_0)^N(1-p_2)^N(1-p_4)^N$.

To bound the probability of incorrectly accepting when the required condition is violated, we therefore consider
\begin{align}\label{eq:stat_pre_test_max}
    \max\quad  &  (1-p_0)^N(1-p_2)^N(1-p_4)^N\\
    \text{s.t.}\quad & \sqrt{p_0+p_2}+\sqrt{p_2+p_4}\geq \frac{1}{2},\nonumber \\
    & p_0,p_2,p_4\in [0,1],\nonumber 
\end{align}
for fixed $N$.
Since the power is a monotonic function, we can equivalently solve the above problem for $N=1$.
Next, we establish that the maximum in \cref{eq:stat_pre_test_max} must satisfy $p_0=p_4$. Let $(p_0,p_2,p_4)$ be a feasible point. Then let $x=\frac{p_0+p_4}{2}$. 
Since the square root is concave, we have $\sqrt{x+p_2}+\sqrt{p_2+x}\geq \sqrt{p_0+p_2}+\sqrt{p_2+p_4}$, i.e., the point $(x,p_2,x)$ is also feasible. At the same time, $(1-x)(1-p_2)(1-x)\geq (1-p_0)(1-p_2)(1-p_4)$, with the inequality being strict if $p_0\neq p_4$.
Having established that $p_0=p_4=x$ leads to the optimal solution, we can resolve the constraint $2\sqrt{p_2+x}\geq \frac{1}{2}$ easily and solve the one-parameter optimization, which leads to $x=0$ and $p_2=\frac{1}{16}$.  
Consequently, it is sufficient to choose $N\geq \log_{\frac{15}{16}}(\delta)$ with the smallest integer satisfying this for $\delta=0.01$ being $72$. 

In the main text, we explain that a cheap preliminary scan over several sequence lengths reveals the coarse damped-oscillation pattern and thereby resolves the aliasing ambiguity associated with potentially larger phase errors when choosing $j$ in \cref{lemma:xi_bound}. In the discussion following \cref{lemma:xi_bound}, we formulate this requirement more explicitly: applying the lemma for a general $j$ requires a preliminary bound $\abs{\theta-\frac{\pi}{2}}\leq\eta$ on the phase $\theta$ of the rotation eigenvalue $\lambda$. This preliminary bound can itself be certified by applying \cref{lemma:xi_bound} with $j=1$, for which no prior phase constraint is required. In the experiment, for the $\tau=1$~ms detection window, this preliminary analysis yields $\abs{\theta-\frac{\pi}{2}}\leq0.08(2)\, \mathrm{rad}$. This is sufficient to justify the choice of $j$ used for the tighter bound reported in \cref{tab:qsq_results}, which requires only $\eta\approx0.10\, \mathrm{rad}$.

\begin{figure*}[t]
	\centering
	\subfloat[\label{methods:fig:simulation_1}]{
		\includegraphics[width=0.48\textwidth]{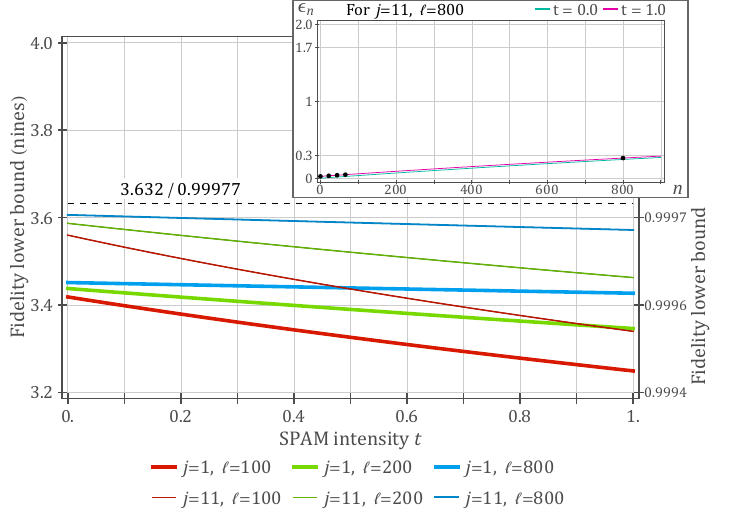}
	}
	\hfill
	\subfloat[\label{methods:fig:simulation_3}]{
		\includegraphics[width=0.48\textwidth]{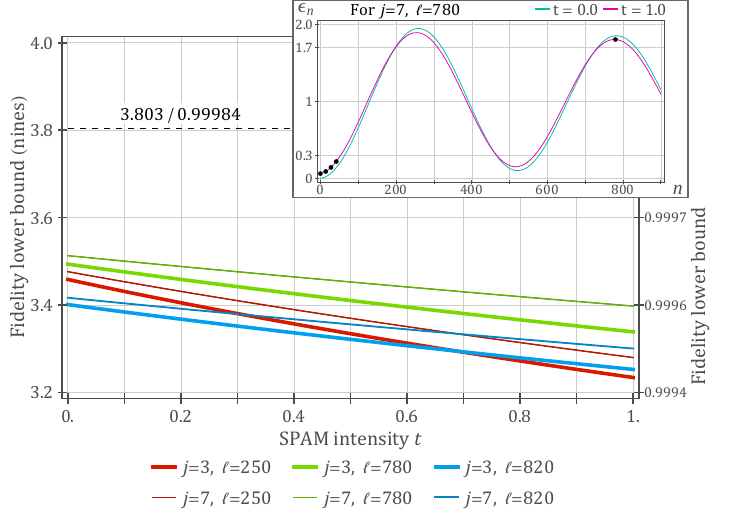}
	}
	\caption{ \textbf{Robustness of the fidelity certificates to increasing \ac{SPAM} errors for different choices of $j$ and $\ell$}. Certified lower bounds on the average gate fidelity, expressed in $\mathrm{Nines}(\Fid_\avg)=-\log_{10}(1-\Fid_\avg)$.  In both panels, the gate channel is fixed but noisy, while $t\in[0,1]$ controls only the \ac{SPAM}-error intensity. Thus, $t=0$ corresponds to ideal \ac{SPAM}, not to an ideal gate.
    The dashed lines (black) indicate the exact fidelity. Insets show the $\veps_n$ signal for two different \ac{SPAM} configurations, $t=0$ (cyan) and $t=1$ (magenta). Black dots are located at $n \in \{0, 2j, 4j, 6j, \ell\}$ chosen according to the tightest lower bounds on the main plot.
    \textbf{(a)} We consider purely dephasing gate noise, with $\vartheta=0$ and $s=0.0007$. State preparation is ideal, while the readout errors are $p_{0|1}=0.01t$ and $p_{1|0}=0.014t$, following the estimated experimental model. 
    \textbf{(b)} As gate noise, we combine coherent over-rotation and dephasing, with $\vartheta=0.012$ and $s=0.0004$. State preparation is given by Bloch vector $\vec{r}$ as in \cref{methods:eq:rho_num_param} with $\alpha = 0.02\pi t$, $\beta=-\pi/2$, $\gamma=0.01t$, corresponding to an off-plane angular displacement and isotropic depolarization, while the readout errors are $p_{0|1}=0.02t$ and $p_{1|0}=0.028t$.  Tighter bounds are generally obtained by choosing the largest $\ell$ for which $\abs{1-\veps_\ell}$ is a local maximum satisfying $\abs{1-\veps_\ell}\geq7/10$. In \textbf{(b)}, for example, $\ell=780$ (green) yields a tighter bound than $\ell=820$ (blue).}
 	\label{methods:fig:simulation_combined}
\end{figure*}

\subsection{Numerical experiments}
\label[methods]{methods:simulations}
In this section, we numerically assess the performance of our certification method across several channel-noise models and increasing levels of \ac{SPAM} error, and examine its dependence on the sequence length parameters $j$ and $\ell$. We work in the infinite-sample limit, using exact probabilities and therefore neglecting statistical fluctuations.

We represent all noise processes using their \ac{PTM} representations (see also \cref{app:eq:state_measurement_param}). We consider two principal gate-error mechanisms: dephasing and coherent over- or under-rotation. Dephasing is described by a channel with the traceless part of the \ac{PTM} being
\begin{equation}
    L{_\mathrm{deph}}\coloneqq 
    \begin{pmatrix}
        1-s&0&0\\
        0&1-s&0\\
        0&0&1
    \end{pmatrix},
\end{equation}
where $s\in[0,1]$ denotes the dephasing strength, with $s=0$ corresponding to noiseless evolution. A coherent rotation error about the pulse axis is described by
\begin{equation}
    L{_\mathrm{rot}}\coloneqq
    \begin{pmatrix}
        1&0&0\\
        0&\cos(\vartheta) & -\sin(\vartheta)\\
        0&\sin(\vartheta) & \cos(\vartheta)
    \end{pmatrix},
\end{equation}
where $\vartheta=0$ corresponds to no coherent rotation error.

The implemented noisy gate is modeled by a channel with the traceless part of the \ac{PTM} being $L = L{_\mathrm{deph}}L{_\mathrm{rot}}L_0$, where $L_0$, defined in \cref{eq:L0}, corresponds to the ideal $\sqrt{\Xgate}$ gate.
Thus, the ideal operation is followed by a coherent rotation error and dephasing. To benchmark our certificates, we use~\cref{app:eq:Favg_Schur} to compute the exact average gate fidelity in the canonical real Schur decomposition gauge, with the required canonicalization described below, and compare it with the fidelity lower bounds produced by our method.

We model \ac{SPAM} errors in the same canonical real Schur decomposition gauge. State-preparation errors are introduced through the parametrized Bloch vector 
\begin{equation}\label{methods:eq:rho_num_param}
    \vec{r} = (1 - \gamma)(\sin{\alpha} \cos{\beta}, \sin{\alpha} \sin{\beta}, \cos{\alpha})^\T
\end{equation}
where $\alpha\in[0,\pi]$ and $\beta\in[0,2\pi)$ determine the direction of the Bloch vector, while $\gamma\in[0,1]$ determines its contraction. The choice $\alpha=\beta=\gamma=0$ corresponds to ideal state preparation in the chosen gauge.  In the same gauge, the readout errors are parametrized by the measurement effect
\begin{equation}
M_0 = (1-p_{1|0})\ketbra{0}{0}+p_{0|1}\ketbra{1}{1},
\end{equation}
where $p_{b|a}\in [0,1]$ denotes the probability of reporting outcome $b$ when the system is in state $\ket{a}$.

To simulate the \ac{SPAM}-robust fidelity estimates for a noisy channel's \ac{PTM}, we first perform a real Schur decomposition on the channel's $L$ submatrix, canonized with eigenvalues $\lambda$ ordered by increasing $|\mathrm{Im}(\lambda)|$ such that $V L V^\T$ is of the form in \cref{app:eq:VLV}, from which we compute the exact fidelity using \cref{app:eq:Favg_Schur}. The corresponding orthogonal transformation $V$ then provides a canonical gauge for the experiment, wherein state preparation and measurements are specified. The simulated probabilities are then used together with \cref{eq:amplitude_bound} and \cref{eq:fid} to obtain certified fidelity lower bounds.

We investigate robustness to \ac{SPAM} in two representative scenarios, with the strength of the \ac{SPAM} errors controlled by $t\in[0,1]$. In both scenarios, the channel of the implemented gate is nonideal and fixed; varying $t$ changes only the state-preparation and measurement errors. The first model considers a gate affected solely by dephasing, together with ideal state preparation and asymmetric readout errors comparable to those observed experimentally. The second model combines dephasing and coherent over-rotation with errors in both state preparation and readout. These include an angular displacement and depolarization of the initial state, as well as stronger asymmetric readout errors. The precise parameters are given in the caption of \cref{methods:fig:simulation_combined}.

\Cref{methods:fig:simulation_combined} compares the resulting fidelity lower bounds for several choices of sequence length parameters $j$ and $\ell$. As we focus on the high-fidelity regime, for increased readability we report fidelities in units of ``nines'', i.e., $\mathrm{Nines}(\Fid_\avg) := -\log_{10}(1 - \Fid_\avg)$, such that $\Fid_\avg = 0.999$ corresponds to 3 nines of fidelity.

Tighter bounds are generally obtained by choosing the largest available $\ell$ for which $\abs{1-\veps_\ell}$ is a local maximum and satisfies the threshold $\abs{1-\veps_\ell}\geq7/10$ required by \cref{lemma:amplitude_bound}. The optimal value of $j$ for a given $\ell$ depends on the channel. Nevertheless, across the scenarios considered here, the resulting bounds vary only modestly over a broad range of admissible choices of $j$ and $\ell$. A coarse preliminary scan of the signal is therefore generally sufficient to select these parameters.

The simulations also illustrate how different \ac{SPAM} mechanisms affect the observed signal. An angular displacement of the initial state within the rotation plane introduces a phase offset in the oscillations of $\veps_n$ and can directly affect the inferred rotation angle. Readout errors predominantly reduce the oscillation contrast, to which the long-sequence certificate remains comparatively robust. Additional noise models and more extensive scans over $j$ and $\ell$ are provided in the accompanying notebooks~\cite{experiment}.

\subsection{Software Package}\label[methods]{methods:software}
Here we describe the software package used to execute certification sequences in both hardware and simulation. QSQlab is a Python package that interfaces with Qiskit at a high level to construct and run quantum circuits \cite{QSQlab, Qiskit}. Here QSQ stands for Quantum System Quizzing, a general name for the deterministic certification protocol introduced by some of us in Ref.~\cite{noller2025classical,noller2025sound}. The package wraps Qiskit's \texttt{Backend} classes through a common interface, so that identical code paths are used for both hardware execution and simulation, ensuring that any comparison between the two reflects genuine physical differences rather than software artifacts. Single-qubit QSQ sequences are generated via the \texttt{BenchmarkSequencer} class, which produces the corresponding parameter and result data as dictionaries; these can optionally be serialized to JSONL files for later analysis. Result dictionaries are then processed by the \texttt{QSQPostProcessing} class, which computes protocols' short- and long-sequence fidelity lower bounds along with their statistical uncertainties, as well as additional indicators related to the~\ac{SPAM} quality.
For simulation, a dedicated wrapper around Qiskit's \texttt{AerSimulator} allows the user to construct noisy simulation backends and execute circuits through the same pipeline used for hardware experiments, enabling direct comparison between simulated and experimental certification protocol indicators. QSQlab is openly available at \cite{QSQlab} under the Apache 2.0 license (version 0.1.0 used in this work), where example workflows and further documentation can be found.

\onecolumngrid
\section*{Appendix}
\g@addto@macro\appendix{%
	\counterwithin{equation}{section}%
	\renewcommand{\theHequation}{\thesection.\arabic{equation}}%
}
\begin{appendix}
\crefalias{section}{appendix}
\crefalias{subsection}{appendix}

In this Appendix, we provide technical details that support the statements in the main text and the Methods section.
In \cref{app:PTM}, we give technical preliminaries such as \acf{PTM} representation of quantum channels, the formula for the average gate fidelities of qubit channels, and constraints on the channels, in particular their spectrum, as well as states and measurements that arise from their physicality, i.e., positivity and complete positivity.
In \cref{app:sound_short}, we give a proof of the short-sequence certificate for the $\pi/2$ quantum model, formalized in \cref{th:sound_short}.
In \cref{app:complete_short}, we state and prove the completeness guarantees for tests with sequence lengths $\Set{0,2,4}$, that is, we prove that models that are close to the target one produce the correct outcomes with high probability.
In \cref{app:sec_gate_SPAM}, we establish the spectral prerequisites for the long-sequence analysis: we show that the short-sequence data certify the presence of one real eigenvalue and one nonreal complex-conjugate pair in the traceless block of the implemented channel, and relate these eigenvalues to the average gate fidelity in a suitable unitary gauge. In \cref{app:long}, we derive certificates on the amplitude and phase of the complex eigenvalue from long-sequence tests and show how additional long-sequence data can sharpen the lower bound on its amplitude. 
In \cref{app:gauge}, we demonstrate that optimization over non-unitary gauge transformations can overestimate the average gate fidelity and analyze how this effect depends on~\ac{SPAM}. Finally, in \cref{app:lemmata}, we collect the auxiliary mathematical results used in the proofs.

\section{Technical preliminaries}
\label{app:PTM}
For the theoretical analysis, it is convenient to work with the vector representation of linear operators on $\H\cong \CC^2$.
In this representation, qubit channels are represented by matrices in $\RR^{4\times 4}$, commonly referred to as \acp{PTM}, if one chooses the set of Pauli matrices $\Set{\1,\X,\Y,\Z}$ as the basis of the operator space.
Let $\Lambda:\LL(\H)\to\LL(\H)$ be a \ac{CPTP} map with $\H\cong\CC^2$, its \ac{PTM}, which we denote as $\hat{\Lambda}$ takes the form
\begin{equation}\label{app:eq:model_param}
    \hat{\Lambda} =
    \begin{pmatrix}
        1 & \begin{matrix}0&0&0\end{matrix}\\
        \vec{l} & L
    \end{pmatrix},
\end{equation}
where $\vec l\in\RR^3$ and $L\in\RR^{3\times 3}$ are defined by
\begin{equation}
    \vec{l}_\sigma=\frac{1}{2}\Tr[\Lambda(\1)\sigma],
    \qquad
    L_{\sigma',\sigma} = \frac{1}{2}\Tr[\Lambda(\sigma)\sigma'],
    \qquad
    \sigma,\sigma'\in\Set{\X,\Y,\Z}.
\end{equation}
Here and throughout the appendix, Pauli matrices are used as indices for the entries of vectors and matrices.

We similarly parametrize the initial state and the binary \ac{POVM} as
\begin{equation}\label{app:eq:state_measurement_param}
    \hat{\rho} =\frac{1}{2}
    \begin{pmatrix}
        1\\ \vec r
    \end{pmatrix},
    \qquad
    \hat{M}_0 = \frac{1}{2}
    \begin{pmatrix}
        1+\mu\\ \vec m
    \end{pmatrix},
    \qquad
    \hat{M}_1 = \frac{1}{2}
    \begin{pmatrix}
        1-\mu\\ -\vec m
    \end{pmatrix},
\end{equation}
where $\mu\in\RR$ and $\vec r,\vec m\in\RR^3$, with
\begin{equation}
    \vec{r}_\sigma=\Tr[\rho\sigma],
    \qquad
    \vec{m}_\sigma=\Tr[M_0\sigma],
    \qquad
    \sigma\in\Set{\X,\Y,\Z}.
\end{equation}

In this parametrization, the target model
$ \bigl(\kb{0}{0},\Set{\sqrt{\Xgate}},\Set{\kb{0}{0},\kb{1}{1}}\bigr)$
takes the form
\begin{equation}\label{app:eq:target_L}
    \vec{r}_0=\vec{m}_0=
    \begin{pmatrix}
        0\\0\\1
    \end{pmatrix},
    \qquad
    \mu_0=0,
    \qquad
    \vec{l}_0=0,
    \qquad
    L_0\coloneqq \begin{pmatrix}
        1&0&0\\
        0&0&-1\\
        0&1&0
    \end{pmatrix}.
\end{equation}

The probability of obtaining outcome $0$ after $n$ applications of the implemented channel is
\begin{equation}\label{app:eq:prob_param}
\begin{split}
    \Pr[0\vert n]=2\hat{M}_0^\T\hat{\Lambda}^n\hat{\rho}=\frac{1+\mu}{2}+\frac{1}{2}\sum_{j=0}^{n-1}\vec m^\T L^j\vec l+\frac{1}{2}\vec m^\T L^n \vec r ,
\end{split}
\end{equation}
with the convention that the sum is empty when $n=0$.

The proofs use sums of failure probabilities over pairs of consecutive even-length tests. For $n\geq 0$, the quantities $\veps_n$ in \cref{eq:epsilon_n} take the form
\begin{equation}\label{app:eq:eps_k_mLq}
    1-\veps_n = (-1)^{n/2}\vec m^\T L^{n}\vec q,
    \qquad
    \vec q
    \coloneqq
    \frac{1}{2}\bigl(\vec r-L^2\vec r-\vec l-L\vec l\bigr).
\end{equation}
The advantage of the paired quantities $\veps_n$ is that the measurement bias $\mu$, as well as the accumulated noise due to non-unitality $\vec{l}$, cancel, as do all affine terms in \cref{app:eq:prob_param} except those contained in $\vec q$. This cancellation is what makes the short-sequence analysis partially robust to readout noise.

We repeatedly use the relation between unitary changes of basis in $\H\cong\CC^2$ and orthogonal transformations of Bloch vectors in $\RR^3$. Every unitary $U\in\U(2)$ induces an orthogonal transformation on the traceless part of the \ac{PTM}. Conversely, every element of $\OO(3)$ arises in this way, up to the global phase of $U$ and possible complex conjugation. In the proofs, we work directly with the corresponding orthogonal matrix $V$ and transform the implemented model as
\begin{equation}
    \vec r\mapsto V\vec r,
    \qquad
    \vec m\mapsto V\vec m,
    \qquad
    \vec l\mapsto V\vec l,
    \qquad
    L\mapsto VLV^\T .
\end{equation}
The associated unitary gauge is denoted by $U$, chosen so that conjugating the target by $U$ corresponds to the orthogonal transformation $V$. 

We now give the fidelity formulas used throughout the proofs. First, the state fidelity with the target preparation in gauge $U$ is
\begin{equation}\label{app:eq:state_fid}
    \Fid(U\ket{0},\rho) = \frac{1}{2}\left(1+\vec{r}_0^\T V\vec r\right).
\end{equation}
Similarly, the distance of the implemented measurement effect from the ideal computational-basis effect in the same gauge is
\begin{equation}\label{app:eq:meas_dist}
    \norm{U\kb{0}{0}U^\dagger-M_0}_\infty=\frac{1}{2}\norm{-\mu\1+\bigl(\vec{m}_0-V\vec m\bigr)\cdot(\X,\Y,\Z)}_\infty.
\end{equation}
For the channel fidelity, in the same gauge, we have 
\begin{equation}\label{app:eq:gate_fid}
    \Fid_\avg(\sqrt{\Xgate}_{\vert U},\Lambda)= \frac{1}{2} + \frac{1}{6}\Tr\left[(V^\T L_0V)^\T L\right] = \frac{1}{2} + \frac{1}{6} \Tr\left[L_0^\T VLV^\T\right].
\end{equation}

While every qubit model can be expressed in terms of the parametrization \cref{app:eq:state_measurement_param}, not every choice of parameters $L,\vec{l},\vec{r},\vec{m},\mu$ results in a valid quantum model. 
Constraints arise from the \emph{physicality} conditions on the model, namely $\rho$ being a density operator, $\Set{M_0,M_1}$ being a \ac{POVM} and $\Lambda$ being a \ac{CPTP} map.
The following constraints, which are repeatedly used in our analysis, are necessary conditions that follow from these physicality conditions;
\begin{enumerate}
    \item $\norm{\vec{r}}\leq 1$, $\norm{\vec{m}}\leq \min\{1+\mu,1-\mu\}\leq 1$, and $\norm{\vec{q}}\leq 1$, where the last bound follows from the fact that $\vec{q}$ is the difference of the Bloch vectors of two states, $\rho$ and $\Lambda^2(\rho)$, divided by $2$;
    \item $\norm{L}_\infty \leq 1$; 
    \item The eigenvalues of $L$ are either $\Set{\lambda_0,\lambda_1,\lambda_2}\subset \RR$, or $\Set{\lambda_0,\lambda_1,\lambda_1^\ast}$, with $\lambda_0\in\RR$ and $\lambda_1\in\CC\setminus \RR$, and in the latter case, it holds that $1+\lambda_0\geq 2\abs{\lambda_1}$ (see Ref.~\cite{wolf2010inverse}).
\end{enumerate}

\section{Short-sequence certification}

\subsection{Robust soundness of short-sequence tests}
\label{app:sound_short}
We prove the short-sequence certificate stated in \cref{th:sound_short}. Throughout the proof, we use the \ac{PTM} notation and gauge conventions introduced in \cref{app:PTM}.

\begin{reptheorem}{th:sound_short}[restated]
Let $\veps_0$ and $\veps_2$ be defined as in \cref{eq:epsilon_n} and assume that
$ \sqrt{\veps_0}+\sqrt{\veps_2}<1$.
Then there exists a unitary $U\in\U(2)$ such that
\begin{equation}\label{app:eq:th_short_gate_fid}
    \Fid_\avg(\sqrt{\Xgate}_{\vert U},\Lambda)
    \geq \frac{1}{3} + \frac{2}{3}
    \sqrt{1-(\sqrt{\veps_0}+\sqrt{\veps_2})^2}.
\end{equation}
Moreover, for the same unitary $U$,
$\Fid(U\ket{0},\rho)\geq 1-\veps_0$,
$\norm{U\kb{0}{0}U^\dagger-M_0}_\infty\leq \veps_0$.
\end{reptheorem}

\begin{proof}
In the parametrization in \cref{app:eq:model_param}, $\veps_0$ and $\veps_2$ take the following form
\begin{equation}\label{app:eq:e0e1}
    \veps_0 = 1-\Pr[0\vert0]+\Pr[0\vert2] = 1-\vec{m}^\T \vec{q},\quad 
    \veps_2 = 1+\Pr[0\vert2]-\Pr[0\vert4] = 1+\vec{m}^\T L^2 \vec{q},
\end{equation}
where $\vec{q} = \frac{1}{2}(\vec{r}-\vec{l}-L\vec{l}-L^2\vec{r})$.

Let us choose an initial unitary gauge such that $(V\vec{m})^\T = \begin{pmatrix}0 & 0 & m\end{pmatrix}$, with $m\in [0,1]$ (since $\norm{\vec{m}}\leq 1$), where $V\in \OO(3)$ is the corresponding orthogonal transformation. 
For convenience, let us introduce the following notation for the entries of the $VL^2V^\T$ matrix and the vector $V\vec{q}$:
\begin{equation}\label{app:eq:L2}
    VL^2V^\T \eqqcolon \begin{pmatrix}
    \begin{matrix} \ast & \ast \\ \ast & \ast \end{matrix} & \vec{t} \\ 
    \vec{y}^\T & z
    \end{pmatrix}, \quad
    V\vec{q} \eqqcolon \begin{pmatrix}\vec{h}\\ q\end{pmatrix},
\end{equation}
where $z,q\in \RR, \vec{y},\vec{t},\vec{h}\in\RR^2$, and by $\ast$ we denote matrix elements which we do not need.
The physicality conditions $\norm{\vec{q}}\leq 1$ and $\norm{VL^2V^\T}_\infty\leq 1$ imply $\norm{\vec{h}}^2+q^2\leq 1$ and $\norm{\vec{y}}^2+z^2\leq 1$, respectively.
First, we show that for $\veps_0\approx 0$ and $\veps_2\approx 0$, which is close to what the target model gives, we have $z\approx -1$.
In the parametrization in \cref{app:eq:L2}, the expressions for $\veps_0$ and $\veps_2$ take the form
\begin{equation}\label{app:eq:e0e1_simp}
  \veps_0=1-mq,\quad \veps_2 = 1+m(\vec{y}^\T\vec{h}+zq),
\end{equation}
from where, using the assumption $\veps_0<1$, we can obtain that 
\begin{equation}\label{app:eq:z_ub}
  z = -\frac{1-\veps_2}{1-\veps_0}-\frac{m}{1-\veps_0}\vec{y}^\T\vec{h}\leq -\frac{1-\veps_2}{1-\veps_0}+\frac{m}{1-\veps_0}\norm{\vec{y}}\norm{\vec{h}}.
\end{equation} 
Using the physicality constraints mentioned above, we can further upper-bound $z$ in \cref{app:eq:z_ub} as
\begin{equation}
  z\leq -\frac{1-\veps_2}{1-\veps_0}+\frac{1}{1-\veps_0}\sqrt{1-z^2}\sqrt{m^2-(1-\veps_0)^2}\leq -\frac{1-\veps_2}{1-\veps_0}+\frac{1}{1-\veps_0}\sqrt{1-z^2}\sqrt{1-(1-\veps_0)^2}\ .
\end{equation}
We can then resolve the above quadratic constraint with respect to $z$, which yields
\begin{equation}\label{app:eq:zUB}
  z\leq -(1-\veps_0)(1-\veps_2)+\sqrt{\left(1-(1-\veps_0)^2\right)\left(1-(1-\veps_2)^2\right)}\leq -1+(\sqrt{\veps_0}+\sqrt{\veps_2})^2,
\end{equation}
which shows that $z\approx -1$, whenever $\veps_0\approx 0$ and $\veps_2\approx 0$.
Since we assume that $\sqrt{\veps_0}+\sqrt{\veps_2}<1$, we have that $z<0$.

Now, let us introduce the following notation for the entries of $VLV^\T$:
\begin{equation}\label{app:eq:L}
  VLV^\T \eqqcolon \begin{pmatrix}
  S & \vec{a} \\ \vec{b}^\T & c
\end{pmatrix},
\end{equation} 
where $c\in\RR$, $\vec{a},\vec{b}\in\RR^2$, $S\in \RR^{2\times 2}$.
Connecting the parametrizations in \cref{app:eq:L2} and \cref{app:eq:L}, we have that $z = c^2+\vec{b}^\T\vec{a}$, and hence the same upper bound in \cref{app:eq:zUB} also holds for $\vec{b}^\T\vec{a}$.
We can now apply an additional unitary gauge corresponding to an orthogonal transformation $W\oplus 1\in \OO(3)$, with $W\in\RR^{2\times 2}$, such that $W\vec{a} = \begin{pmatrix}0\\ -a\end{pmatrix}$, for some $a\in \Rnn$.
Let $(W\vec{b})^\T = \begin{pmatrix}\ast & b\end{pmatrix}$ for $b\in \RR$, which gives $\vec{b}^\T\vec{a} = -ab$.
Importantly, this additional unitary gauge does not transform the subspace corresponding to the Pauli-$\Z$ operator, and thus it does not affect the bound in \cref{app:eq:zUB}.

Now, we show that $ab\approx 1$ implies $\Fid_\avg(\sqrt{\Xgate}_{\vert U},\Lambda)\approx 1$ for $U^\dagger$ corresponding to $(W\oplus 1)V$, i.e. for the constructed gauge.
We use the following formula for the average gate fidelity between two channels $\Lambda$ and $\Lambda'$ in terms of their Pauli transfer matrices
\begin{equation}
    \Fid_\avg(\Lambda',\Lambda) = \frac{1}{3}+\frac{1}{6}\Tr[\hat{\Lambda}'^\T\hat{\Lambda}],
\end{equation}
which holds whenever one of the channels is unital (see e.g., Ref.~\cite{kliesch2021theory}).
This gives us the following expression for $\Fid_\avg(\sqrt{\Xgate}_{\vert U},\Lambda)$ for the chosen unitary gauge $U$ in the parametrization in \cref{app:eq:L} (remember the form of the Pauli transfer matrix for the target gate in \cref{app:eq:target_L}),
\begin{equation}
      \Fid_\avg(\sqrt{\Xgate}_{\vert U},\Lambda) = \frac{1}{3}+\frac{1}{6}\left(1+\Tr\left[
    \begin{pmatrix} 1 & 0\\ 0 & 0
  \end{pmatrix} S\right]+a+b\right).
\end{equation}
Now, we use the condition that $\Lambda$ is \ac{CP}, which, in particular, implies that the entanglement fidelity $\frac{1}{4}\Tr[\hat{\Lambda'}^\T\hat{\Lambda}]$ between $\Lambda$ and any other channel $\Lambda'$, must be nonnegative.
Taking $\Lambda'=\sqrt{\Xgate}^\dagger(\cdot)\sqrt{\Xgate}$, for which the Pauli transfer matrix is clearly just the transpose of the one for $\sqrt{\Xgate}$-gate, gives us
\begin{equation}
      1+\Tr\left[\begin{pmatrix}
    1 & 0\\ 0 & 0
  \end{pmatrix} S\right] - a - b \geq 0\ ,
\end{equation}
from where we directly get a simple lower bound $\Fid_\avg(\sqrt{\Xgate}_{\vert U},\Lambda)\geq \frac{1}{3}+\frac{1}{3}(a+b)$.
Finally, we use the relation $a+b\geq 2\sqrt{ab}$, along with $ab = -\vec{b}^\T\vec{a}\geq -z$, and the upper bound in \cref{app:eq:zUB} to obtain
\begin{equation}
    \Fid_\avg(\sqrt{\Xgate}_{\vert U},\Lambda)\geq \frac{1}{3}+\frac{2}{3}\sqrt{1-(\sqrt{\veps_0}+\sqrt{\veps_2})^2}\ .
\end{equation}
The claims of the theorem for the state $\rho$ and the \ac{POVM} effect $M_0$ are straightforward. 
From \cref{app:eq:e0e1_simp}, since $q\leq 1$, we have that $m\geq 1-\veps_0$, and thus 
\begin{equation}
    \norm{U\kb{0}{0}U^\dagger-M_0}_\infty = \frac{1}{2}\norm{-\mu\1+(1-m)\Z}_\infty=\frac{1}{2}\max\Set{1-m-\mu,1-m+\mu}\leq 1-m\leq \veps_0,
\end{equation}
where we used the fact that $m\leq \min\Set{1+\mu,1-\mu}$.
As for the state $\rho$, we use the fact that $1-\veps_0=m\vec{q}_\Z$, leading to $\vec{q}_\Z\geq 1-\veps_0$, since $m\leq 1$.
However, from the definition of $\vec q$, we have $\vec{r}_\Z=2\vec{q}_\Z +(\vec{l}+L\vec{l}+L^2\vec{r})_\Z\geq 2(1-\veps_0)-1=1-2\veps_0$.
Therefore, $\Tr[U\kb{0}{0}U^\dagger\rho]=\frac{1+\vec{r}_\Z}{2}\geq 1-\veps_0.$
\end{proof}

\subsection{Robust completeness of the short-sequence test}
\label{app:complete_short}

We now prove that models close to the target model pass the short deterministic checks with high probability.

\begin{theorem}\label{th:complete_short}
Let a qubit quantum model $(\rho,\Set{\Lambda},\Set{M_0,M_1})$ be close to the target model in the sense that there exists $U\in\U(2)$ such that
\begin{equation}\label{app:eq:th_complete}
\begin{split}
    \Fid_\avg(\Xgate_{\vert U},\Lambda) &\geq 1-\veps,\\
    \Fid(U\ket{0},\rho) &\geq 1-\veps_{\rho},\\
    \norm*{U\kb{0}{0}U^\dagger-M_0}_\infty &\leq \veps_M .
\end{split}
\end{equation}
Assume $\veps_\rho\leq 1/20$ and $\veps\leq 1/27$. Then
\begin{equation}\label{app:eq:complete_probs}
\begin{split}
    \Pr[0\vert 0] &\geq 1-\veps_\rho-\veps_M,\\
    \Pr[1\vert 2] &\geq (1-12\veps)(1-2\veps_\rho)-\veps_M,\\
    \Pr[0\vert 4] &\geq
    (1-12\veps)\bigl(2(1-12\veps)(1-2\veps_\rho)-1\bigr)-\veps_M .
\end{split}
\end{equation}
\end{theorem}

\begin{proof}
It is clear that in the proof of the above claim, one can always fix $U=\1$, because otherwise, one can simply redefine the implemented model, i.e., $\rho\mapsto U\rho U^\dagger$, etc., and work with that.

First, we use the third constraint in \cref{app:eq:th_complete} to lower-bound the probability of obtaining the correct outcome $a$ for any tested sequence of length $k \in \Set{0,2,4}$ as
\begin{equation}
  \Pr[a\vert k] = \Tr[\kb{a}{a}\Lambda^k(\rho)]-\Tr[(\kb{a}{a}-M_a)\Lambda^k(\rho)]\geq \Tr[\kb{a}{a}\Lambda^k(\rho)]-\veps_M,
\end{equation}
with the convention that $\Lambda^0$ is the identity channel.
In the following, we directly lower-bound $\Tr[\kb{a}{a}\Lambda^k(\rho)]$ for all $k\in\{0,2,4\}$, and subtract $\veps_M$ at the end.
A lower bound on the term $\Tr[\kb{0}{0}\rho]$, corresponding to the case $k=0$, immediately follows from the second inequality in \cref{app:eq:th_complete} (since we took $U=\1$), and thus,
\begin{equation}\label{app:eq:p0_lb}
  \Pr[0\vert 0]\geq 1-\veps_\rho-\veps_M.
\end{equation}
Same as in \cref{app:sound_short}, we parametrize the implemented model in the basis of Pauli matrices, and this time, use the following notation for the entries of the corresponding vectors and matrices
\begin{equation}\label{app:eq:model_param_complete}
\hat{\rho} = \frac{1}{2}\begin{pmatrix}1\\ \vec{g} \\ r\end{pmatrix},\quad 
\hat{\Lambda} = \begin{pmatrix}1 & \begin{matrix} 0 & 0 \end{matrix} & 0\\ \begin{matrix}\ast \\ \ast \end{matrix} & S & \vec{a} \\ \ast & \vec{b}^\T & c \end{pmatrix},\quad 
\hat{\Lambda}^2 = \begin{pmatrix} 
1 & \begin{matrix} 0 & 0 \end{matrix} & 0 \\
\ast & \begin{matrix} \ast & \ast \end{matrix} & \ast \\
\ast & \begin{matrix} \ast & \ast \end{matrix} & \ast \\ 
l &  \vec{y}^\T & z\end{pmatrix},
\end{equation}
with $r,c,l,z\in \RR$, $\vec{a},\vec{b},\vec{y},\vec{g}\in\RR^2$, and $S\in\RR^{2\times 2}$, where we again denote by $\ast$ the entries of the matrices that we do not need for the proof.
The proof again revolves around bounding the element $z$ of the channel $\Lambda^2$. 
From the bound on the average gate fidelity (the first condition in \cref{app:eq:th_complete}), we get
\begin{equation}
  1-\veps\leq \Fid_\avg(\sqrt{\Xgate},\Lambda) = \frac{1}{2}+\frac{1}{6}\Tr\left[
  \begin{pmatrix}
    1 & 0 & 0\\
    0 & 0 & -1\\
    0 & 1 & 0
  \end{pmatrix}^\T
  \begin{pmatrix}
    S & \vec{a}\\
    \vec{b}^\T & c
  \end{pmatrix}
  \right]\ ,
\end{equation}
and hence 
\begin{equation}
    \Tr\left[ \begin{pmatrix} 1 & 0\\ 0 & 0 \end{pmatrix}S\right] + \begin{pmatrix}0 & 1\end{pmatrix}(\vec{b}-\vec{a})\geq 3-6\veps.
\end{equation}
At the same time, from the positivity of $\Lambda$, we know that $\norm{S}_\infty \leq 1$, which implies that $\begin{pmatrix}0 & 1\end{pmatrix}(\vec{b}-\vec{a})\geq 2-6\veps$.
Also from the positivity of $\Lambda$, we have that $\norm{\vec{a}}^2+c^2\leq 1$ and $\norm{\vec{b}}^2+c^2\leq 1$. 
Now we bound $\norm{\vec{b}-\vec{a}}$ from above and below, namely
\begin{equation}
  \begin{split}
    \norm{\vec{b}-\vec{a}} &\geq \abs{\begin{pmatrix}0 & 1\end{pmatrix}(\vec{b}-\vec{a})}\geq 2-6\veps\ ,\\
    \norm{\vec{b}-\vec{a}}^2 &= \norm{\vec{b}}^2+\norm{\vec{a}}^2-2\vec{b}^\T\vec{a}\leq 2-2c^2-2\vec{b}^\T\vec{a} = 2-2z\ ,
  \end{split}
\end{equation}
from where we conclude that $z\leq -1 + 12\veps - 18\veps^2\leq -1+12\veps$.

Now we use this upper bound on $z$ to lower-bound $\Pr[1\vert 2]$.
In the notation in \cref{app:eq:model_param_complete}, the \ac{POVM} effect ${\Lambda^\dagger}^2(\kb{1}{1})$ is parametrized as $\frac{1}{2}\begin{pmatrix}1-l & -\vec{y}^\T & -z\end{pmatrix}^\T$, where $\Lambda^\dagger$ is the adjoint of $\Lambda$.
The positivity condition for that \ac{POVM} effect implies that $\norm{\vec{y}}\leq \sqrt{(1-l)^2-z^2}$, and the positivity of the state $\rho$ implies that $\norm{\vec{g}}\leq \sqrt{1-r^2}$, from where we obtain that
\begin{equation}\label{app:eq:11Lambda2rho_lb1}
  2\Tr[{\Lambda^\dagger}^2(\kb{1}{1})\rho] = 1-l-\vec{y}^\T\vec{g}-zr\geq 1-l-\sqrt{(1-l)^2-z^2}\sqrt{1-r^2}-zr\ . 
\end{equation}
Now we use a simple relation $x-\sqrt{x^2-z^2}k\geq -z\sqrt{1-k^2}$, which holds for $x\in \Rnn$, $0\leq k\leq 1$, and $\abs{z}\leq x$, along with $r=2\Tr[\rho\kb{0}{0}]-1\geq 1-2\veps_\rho$ to conclude that
\begin{equation}\label{app:eq:11Lambda2rho_lb2}
  2\Tr[\kb{1}{1}\Lambda^2(\rho)] = 2\Tr[{\Lambda^\dagger}^2(\kb{1}{1})\rho] \geq -2z(1-2\veps_\rho)\geq 2(1-12\veps)(1-2\veps_\rho)\ ,
\end{equation}
and thus
\begin{equation}\label{app:eq:p2_lb}
  \Pr[1\vert 2]\geq (1-12\veps)(1-2\veps_\rho)-\veps_M.
\end{equation}

To lower-bound the probability $\Pr[0\vert 4]$, first we argue about the state $\Lambda^2(\rho)$, for which we introduce the following notation
\begin{equation}
  \hat{\Lambda}^2\hat{\rho}= 
  \frac{1}{2}\begin{pmatrix}
  1\\ \vec{f}\\ l+\vec{y}^\T\vec{g}+zr 
  \end{pmatrix}\ ,
\end{equation}
where we introduced a new vector $\vec{f}\in \RR^2$.
From \cref{app:eq:11Lambda2rho_lb1,app:eq:11Lambda2rho_lb2}, we know that $1-l-\vec{y}^\T\vec{g}-zr\geq -2z(1-2\veps_\rho)$, which for $z\approx -1$ implies that $\Lambda^2(\rho)$ is close to the state $\kb{1}{1}$, as we would expect.
From the physicality of $\Lambda^2(\rho)$ we establish that $\norm{\vec{f}}^2\leq 1-(2z(1-2\veps_\rho)+1)^2$,
which follows from the above relation $l+\vec{y}^\T\vec{g}+zr\leq 1+2z(1-2\veps_\rho)$, and $1+2z(1-2\veps_\rho)\leq 0$, which is the case for $\veps_\rho\leq 1/20$ and $\veps\leq 1/27$ as in the statement of the theorem.
Similarly to the previous case of the sequence of length $2$, we bound the following expression
\begin{equation}
  2\Tr[\kb{0}{0}\Lambda^4(\rho)] = \begin{pmatrix}1+l & \vec{y}^\T & z\end{pmatrix}\begin{pmatrix}
  1\\ \vec{f}\\ l+\vec{y}^\T\vec{g}+z r
  \end{pmatrix} \geq 1+l-\sqrt{(1+l)^2-z^2}\norm{\vec{f}}+z(2z(1-2\veps_\rho)+1)\ ,
\end{equation}
where we used that $z<0$, and $\norm{\vec{y}}\leq \sqrt{(1+l)^2-z^2}$, because ${\Lambda^\dagger}^2(\kb{0}{0})$ is a \ac{POVM} effect, and hence, positive semidefinite. 
We use the same relation again as we did in \cref{app:eq:11Lambda2rho_lb2} for the first two terms, which results in
\begin{equation}\label{eq:00Lambda4rho_lb}
  2\Tr[\kb{0}{0}\Lambda^4(\rho)]\geq -z\abs{2z(1-2\veps_\rho)+1}+z(2z(1-2\veps_\rho)+1)=2z(2z(1-2\veps_\rho)+1)\ ,
\end{equation}
where we again use that $2z(1-2\veps_\rho)+1\leq 0$.
Additionally, it holds that $z<-\frac{1}{4(1-2\veps_\rho)}$, and therefore the minimum of the expression in \cref{eq:00Lambda4rho_lb} is attained at the highest value of $z$ that it can take.
Hence, to establish the lower bound on $\Tr[\kb{0}{0}\Lambda^4(\rho)]$, we need to use the upper bound on $z$ that we derived previously.  
This results in
\begin{equation}
  \Tr[\kb{0}{0}\Lambda^4(\rho)]\geq (1 - 12\veps)(2(1 - 12\veps)(1-2\veps_\rho)-1),
\end{equation}
and therefore
\begin{equation}\label{eq:p4_lb}
  \Pr[0\vert 4]\geq (1 - 12\veps)(2(1 - 12\veps)(1-2\veps_\rho)-1)-\veps_M.
\end{equation}
\end{proof}

\section{Prerequisites for long-sequence tests}~\label{app:sec_gate_SPAM}
Here we present the proof of~\cref{lemma:eigenvalues}, which we restate below.

\begin{replemma}{lemma:eigenvalues}[restated]
    Let $\veps_0$, $\veps_2$, as defined in \cref{eq:epsilon_n}, satisfy $\sqrt{\veps_0}+\sqrt{\veps_2}<\frac{1}{2}$.
    Then, the block $L$ in the~\ac{PTM} of the implemented channel $\Lambda$  is diagonalizable, with eigenvalues $\Set{\lambda_0,\lambda_1,\lambda_1^\ast}$, where $\lambda_0\in\RR$, and $\lambda_1\in\CC\setminus\RR$.
\end{replemma}
 
\begin{proof}
The eigenvalues of $L$, the traceless part of the $\ac{PTM}$ of $\Lambda$, are either all real or its complex eigenvalues appear with their complex conjugates in the spectrum of $L$~\cite{wolf2010inverse}, that is $\Set{\lambda_0,\lambda_1,\lambda_1^\ast}$ with $\lambda_0\in \RR$.
Here we show that if $\sqrt{\veps_0}+\sqrt{\veps_2}<\frac{1}{2}$, then $L$ cannot have all its eigenvalues real.

Let us consider the same parametrization of $L$ as in \cref{app:eq:L} and $L^2$ as in \cref{app:eq:L2}, i.e., we consider the same orthogonal transformations $V$ and $W\oplus 1$ applied to $L$ as in \cref{th:sound_short}.
From the proof of \cref{th:sound_short}, we know that $z\leq -1+(\sqrt{\veps_0}+\sqrt{\veps_2})^2$ (see \cref{app:eq:zUB}).
The key idea is to look at $\Tr[L^2]$, which for the case of all real eigenvalues of $L$ must be non-negative.
At the same time, we have
\begin{equation}
    \Tr[L^2] = \Tr[S^2]+2\vec{b}^\T\vec{a}+c^2 \leq \Tr[S^2]+2z.
\end{equation}
Now, let us introduce the parametrization of the matrix $S$ (in the basis of $W$),
\begin{equation}
    WSW^\T\eqqcolon \begin{pmatrix}
        s_1 & s_2 \\ s_3 & s_4
    \end{pmatrix}.
\end{equation}
From the condition $\norm{L}_\infty\leq 1$, we know that $s_3^2+s_4^2+a^2\leq 1$ and $s_2^2+s_4^2+b^2\leq 1$, where with $a$ and $b$ we denote the same elements of $W\vec{a}$ and $W\vec{b}$ as in the proof of \cref{th:sound_short} (introduced after \cref{app:eq:L}). 
Together with the trivial bound $s_1^2\leq 1$, and $2s_2s_3\leq s_2^2+s_3^2$, this lets us conclude that 
\begin{equation}
    \Tr[S^2]\leq 1+s_2^2+s_3^2+s_4^2\leq 3-a^2-b^2\leq 3-2ab=3+2\vec{b}^\T\vec{a}\leq 3+2z.
\end{equation}
Finally, from the condition of the lemma, $\sqrt{\veps_0}+\sqrt{\veps_2}<\frac{1}{2}$, we have $z<-\frac{3}{4}$, and thus $\Tr[L^2]\leq 3+4z<0$, which proves that $L$ cannot have all its eigenvalues real. 
\end{proof}

\begin{replemma}{lemma:fidelity}[restated]\hypertarget{app:lemma:fidelity}{}
    Let $\veps_0$, $\veps_2$, as defined in \cref{eq:epsilon_n}, satisfy $\sqrt{\veps_0}+\sqrt{\veps_2}<\frac{1}{2}$. Then there exists $U\in \U(2)$ for which \begin{equation}\label{app:eq:methods_fid_lambdas}
    \Fid_\avg(\sqrt{\Xgate}_{\vert U},\Lambda)
\geq \frac12+\frac16\bigl(\lambda_0+2\Im[\lambda_1]\bigr),
\end{equation}
where $\Set{\lambda_0,\lambda_1,\lambda_1^\ast}$ are the eigenvalues of the block $L$ of the \ac{PTM} of the implemented channel $\Lambda$.
\end{replemma}

\begin{proof}
From \cref{lemma:eigenvalues}, we know that $L$ is diagonalizable and has eigenvalues $\Set{\lambda_0,\lambda_1,\lambda^\ast_1}$.
Via the real Schur decomposition, $L$ can be brought to the form
\begin{equation}\label{app:eq:VLV}
    \tilde VL\tilde V^\T = \begin{pmatrix}
        \lambda_0 & \ast & \ast \\
         0 & a & -b \\
         0 & c & a
        \end{pmatrix},
\end{equation}
where $a,b,c\in\RR$, $\lambda_1=a+\i\sqrt{bc}$, with $bc\geq 0$, since we know that it must be that $\lambda_1\notin\RR$, and $\tilde V\in \OO(3)$ (see e.g., Ref.~\cite{Horn_Johnson_1985}).
Thus, by choosing unitary $U$ corresponding to $\tilde V^\T$, we obtain
\begin{equation}\label{app:eq:Favg_Schur}
    \Fid_\avg(\sqrt{\Xgate}_{\vert U},\Lambda) = \frac{1}{2}+\frac{1}{6}\Tr[L_0^\T \tilde VL\tilde V^\T]=\frac{1}{2}+\frac{1}{6}(\lambda_0+b+c)\geq \frac{1}{2}+\frac{1}{6}(\lambda_0+2\Im[\lambda_1]),
\end{equation}
where we used $b+c\geq 2\sqrt{bc}$, and also assumed that $b+c\geq 0$, which can be ensured by the choice of $\tilde{V}$. 
This finishes the proof.
\end{proof}

\section{Eigenvalue certification from long-sequence tests}\label{app:long}
Next, we show how to obtain certificates on $\lambda_1$ from the long-sequence tests.

\begin{replemma}{lemma:xi_bound}[restated, generalized]\hypertarget{app:lemma:xi_bound}{}
    Assume that the eigenvalues of the block $L$ of the \ac{PTM} of the implemented channel $\Lambda$ in \cref{app:eq:model_param} are $\Set{\lambda_0,\lambda_1,\lambda_1^\ast}$,
    where $\lambda_0\in\RR$ and $\lambda_1\in \CC\setminus \RR$. 
    Let $j,i\in\Znn$ with $j$ odd and $i$ even, and let $\veps_{i}^\uparrow\leq \veps_{2j+i}^\uparrow\leq \veps_{4j+i}^\uparrow\leq \veps_{6j+i}^\uparrow$ be a reordering of $\Set{\veps_{i},\veps_{2j+i},\veps_{4j+i},\veps_{6j+i}}$, which are defined in \cref{eq:epsilon_n}. 
    Assume $\veps_{6j+i}^\uparrow<1$, and define
    \begin{equation}\label{app:eq:xi_definition}
        \xi_{i,j}\coloneqq \sqrt{\frac{\veps^\uparrow_{6j+i} + \veps_{4j+i}^
\uparrow-\veps_{2j+i}^\uparrow-\veps_{i}^\uparrow}{1-\veps_{2j+i}^\uparrow}}\ .
    \end{equation}
Then, if $\xi_{i,j}\leq 1$, the complex eigenvalue $\lambda_1=\e^{\ii\theta}\abs{\lambda_1}$ satisfies $\abs{1+\lambda_1^{2j}}\leq \xi_{i,j}$, which implies    \begin{equation}\label{app:eq:amplitude_angle_bound}
        \abs{\lambda_1}\geq (1-\xi_{i,j})^{\frac{1}{2j}},\;\text{ and }\;\theta+ k\frac{\pi}{j}\in\left[\frac{\pi}{2}-\frac{\arcsin(\xi_{i,j})}{2j},\frac{\pi}{2}+\frac{\arcsin(\xi_{i,j})}{2j}\right],\; \text{ for some }k\in\ZZ.
    \end{equation}
Moreover, if there exists $\eta\in [0,\frac{\pi}{j}]$, such that $\theta\in\left[\frac{\pi}{2}-\eta,\frac{\pi}{2}+\eta\right]$ and $\arcsin(\xi_{i,j})<2\pi -2j\eta$, then we can conclude that \cref{app:eq:amplitude_angle_bound} holds for $k=0$. 
\end{replemma}

\begin{proof}
Set $i$ and $j$ as in the statement of the lemma.
Define $T \coloneqq L^{2j}$, whose eigenvalues are $\lambda_0^{2j}\eqqcolon \nu_0$, $\lambda_1^{2j}\eqqcolon \nu_1$ and $\nu_1^\ast$.
In terms of $T$, the quantities $\veps_{2kj+i}$ in \cref{eq:epsilon_n} for $k\in\Set{0,1,2,3}$ take the form
\begin{equation}\label{app:eq:eps_jki}
    1-\veps_{2kj+i} = (-1)^{kj+\frac{i}{2}}\vec{m}^\T L^{2k j+i}\vec{q} = (-1)^{k+\frac{i}{2}}\vec{m}^\T T^kL^{i}\vec{q},
\end{equation}
where we used the fact that $j$ is odd.
Let us express the characteristic polynomial of $T$ as
\begin{equation}\label{app:eq:char_poly}
    (t-\nu_0)(t-\nu_1)(t-\nu_1^\ast)\eqqcolon \sum_{k=0}^3c_kt^k,
\end{equation}
where $(c_0,c_1,c_2,c_3)=(-\nu_0\abs{\nu_1}^2,\abs{\nu_1}^2+2\nu_0\Re[\nu_1],-\nu_0-2\Re[\nu_1],1)$.
By the Cayley–Hamilton theorem~\cite{Frobenius_1878} (see also Ref.~\cite{wolf2010inverse}), if we insert $T$ instead of $t$ in the characteristic polynomial in \cref{app:eq:char_poly}, we obtain $\sum_{k=0}^3c_kT^k = 0$. 
We use this fact together with the form of $\veps_{2kj+i}$ in \cref{app:eq:eps_jki} to conclude the following
\begin{equation}\label{app:eq:eps-char_poly_condition}
    \sum_{k=0}^{3}c_k(-1)^{k}(1-\veps_{2kj+i}) = (-1)^{i/2}\sum_{k=0}^{3}c_k\vec{m}^\T T^{k}L^{i}\vec{q}=(-1)^{i/2}\vec{m}^\T\left(\sum_{k=0}^3c_kT^k\right)L^{i}\vec{q} = 0.
\end{equation}
We use the above constraint to make conclusions about $\nu_1$.
First, observe that $\Re[\nu_1]\leq 0$ must hold, or \cref{app:eq:eps-char_poly_condition} cannot be satisfied.
Indeed, since $\lambda_0\in\RR$, and thus $\nu_0\geq 0$, $\Re[\nu_1]>0$ would imply $(-1)^kc_k<0$ for $k\in\Set{1,2,3}$, which together with our assumption on $1-\veps_{2kj+i}$ being positive, contradicts \cref{app:eq:eps-char_poly_condition}.
Having established that $\Re[\nu_1]\leq 0$, we can estimate $c_1$ and $c_2$ from below and above as follows
\begin{equation}\begin{split}
    -1\leq \abs{\nu_1}^2-2\abs{\nu_1}\leq c_1 & =\abs{\nu_1}^2+2\nu_0\Re[\nu_1]\leq\abs{\nu_1}^2\leq 1,\\
    -1\leq 2(\abs{\nu_1}+\Re[\nu_1])-1\leq -c_2 & = \nu_0+2\Re[\nu_1]\leq\nu_0\leq 1,
\end{split}\end{equation}
where we used the physicality constraints $\nu_0\geq 2\abs{\nu_1}-1$, and $\nu_0,\abs{\nu_1}\in [0,1]$ (see \cref{app:PTM}).
Hence, we conclude that $\abs{c_k}\leq1$ for all $k\in \Set{0,1,2,3}$.

Now let us introduce and additional parameter $x\in [0,1)$, and express the equality in \cref{app:eq:eps-char_poly_condition} as
\begin{equation}\label{app:eq:eps-c_k-condition-x}
    -(1-x)\sum_{k=0}^{3}c_k(-1)^{k}=\sum_{k=0}^3 c_k (-1)^k (x-\veps_{2kj+i}).
\end{equation}
Since we established that $\abs{c_k}\leq 1$, we can upper-bound the right-hand side of \cref{app:eq:eps-c_k-condition-x} by $\sum_{k=0}^{3}\abs{x-\veps_{2kj+i}}$.
At the same time, the left-hand side of \cref{app:eq:eps-c_k-condition-x} simplifies as $(1-x)(1+\nu_0)\abs{1+\nu_1}^2$.
Choosing $x=\veps^\uparrow_{2j+i}$ for the re-ordering $\veps^\uparrow_i\leq \veps^\uparrow_{2j+i}\leq\veps^\uparrow_{4j+i}\leq\veps^\uparrow_{6j+i}$, and using a trivial bound $\nu_0\geq 0$, leads to the following bound on $\lambda_1^{2j}=\nu_1$,
\begin{equation}\label{eq:lambda_1-square_condition}
    \abs{1+\lambda_1^{2j}}\leq\sqrt{\frac{\veps^\uparrow_{6j+i} + \veps_{4j+i}^
\uparrow-\veps_{2j+i}^\uparrow-\veps_{i}^\uparrow}{1-\veps_{2j+i}^\uparrow}}=\xi_{i,j} ,
\end{equation}
where the right-hand side is the quantity defined in \cref{app:eq:xi_definition}.

Now we use the constraint in \cref{eq:lambda_1-square_condition} to derive restrictions on both the amplitude and the phase of $\lambda_1$. 
The constraint on the amplitude follows immediately from the triangle inequality
\begin{align}
    1-\abs{\lambda_1}^{2j}\leq \abs{1+\lambda_1^{2j}}\leq\xi_{i,j}\quad&\Rightarrow\quad\abs{\lambda_1}\geq(1-\xi_{i,j})^{\frac{1}{2j}},
\end{align}
for the case $\xi_{i,j}\leq 1$.
As for the phase $\theta$ of $\lambda_1$, first, we can conclude that $\abs{\sin(2j\theta)}\leq \xi_{i,j}$ (see \cref{lemma:phi}), and since $\Re[\lambda_1^{2j}]\leq 0$, we reach the result in \cref{app:eq:amplitude_angle_bound}, that there exists $k\in \ZZ$, for which
\begin{equation}\label{eq:phase_j_intervals}
    \theta+k\frac{\pi}{j} \in \left[\frac{\pi}{2}-\frac{\arcsin(\xi_{i,j})}{2j},\frac{\pi}{2}+\frac{\arcsin(\xi_{i,j})}{2j}\right].
\end{equation}
Finally, if there exists $\eta\geq 0$ for which we know that $\theta\in\left[\frac{\pi}{2}-\eta,\frac{\pi}{2}+\eta\right]$, then clearly under the condition 
\begin{equation}
    \frac{\pi}{2}-\frac{\pi}{j}+\frac{\arcsin(\xi_{i,j})}{2j}<\frac{\pi}{2}-\eta,\quad \text{ and }\quad \frac{\pi}{2}+\eta < \frac{\pi}{2}+\frac{\pi}{j}-\frac{\arcsin(\xi_{i,j})}{2j},
\end{equation}
which translates to $\eta<\frac{\pi}{j}-\frac{\arcsin(\xi_{i,j})}{2j}$, there is no ambiguity for which $k$ the inclusion in \cref{app:eq:amplitude_angle_bound} holds.
\end{proof}

In the next lemma, we show that having some initial trust about the eigenvalue $\lambda_1$ of $L$ in \cref{app:eq:model_param} lets us derive an even tighter lower bound on the absolute value of $\lambda_1$ whenever we observe that a quantity $\veps_n$ in \cref{eq:epsilon_n} is away from $1$ for high $n$.
The intuition here is that for high $n$, due to decoherence, $L^{n}$ tends to the zero matrix and thus $\veps_n$ tends to $1$.

\begin{replemma}{lemma:amplitude_bound}[restated]
    Assume that the block $L$ of the \ac{PTM} of the implemented channel $\Lambda$ in \cref{app:eq:model_param} is diagonalizable with eigenvalues $\{\lambda_0,\lambda_1,\lambda_1^\ast\}$, where $\lambda_0\in\RR$ and $\lambda_1\in\CC\setminus\RR$.
    Assume, moreover, that the complex eigenvalue $\lambda_1=\abs{\lambda_1}\e^{\i\theta}$ satisfies $\abs{\lambda_1}\geq 1-\eta$ for $\eta\leq \frac{1}{10}$ and $\abs{\theta-\frac{\pi}{2}}\leq \frac{\pi}{10}$.
    Then, if $\abs{1-\veps_n}\geq 7\eta$, for an even $n$,  where $\veps_n$ is defined in \cref{eq:epsilon_n}, we have
    \begin{equation}\label{app:eq:abs_lambda_1_lb}
        \abs{\lambda_1}^{n}\geq \frac{\abs{1-\veps_n}-2(1-\abs{\lambda_1})}{\sqrt{\frac{2}{1-8\abs{\lambda_1}(1-\abs{\lambda_1})}-1}}\ .
    \end{equation}  
    Moreover, if $\abs{1-\veps_n}\geq 7/10$, and $n\geq 22$, then $\abs{\lambda_1}^{n-11}\geq \abs{1-\veps_n}$.
\end{replemma}

\begin{proof}
Let $n$ be as in the statement of the lemma. Then
\begin{equation}\label{app:eq:abs_1-veps}
    \abs{1-\veps_n} = \abs{\vec{m}^\T L^{n}\vec{q}}\leq \norm{L^{n}\vec{q}}\leq\frac{1}{2}\norm{L^{n}(\vec{r}-L^2\vec{r})}+\frac{1}{2}\norm{L^{n}(L\vec{l}+\vec{l})}\leq\frac{1}{2}\norm{L^{n}(\1-L^2)}_\infty+\norm{\vec{l}},
\end{equation}
where we used that $\vec{q}= \frac{1}{2}(\vec{r}-L^2\vec{r}-\vec{l}-L\vec{l})$, and the physicality conditions $\norm{L}_\infty\leq 1$, $\norm{\vec{m}}\leq 1$, $\norm{\vec{r}}\leq 1$.
Since we assume that $\abs{\lambda_1}\geq \frac{9}{10}$, and thus $\lambda_0\geq \frac{4}{5}$ from the physicality constraint $\lambda_0\geq 2\abs{\lambda_1}-1$, $L$ is invertible, and we use \cref{lemma:non-unitality_bound} to conclude $\norm{\vec{l}}\leq 2(1-\abs{\lambda_1})$.
By our assumption, $L$ is diagonalizable. 
Let $A$ be such that $L = A\diag{(\lambda_0,\lambda_1,\lambda_1^\ast)}A^{-1}$.
We can then upper-bound $\norm{L^{n}(\1-L^2)}_\infty$ as
\begin{equation}\label{app:eq:Ln(1-L)}
    \norm{L^{n}(\1-L^2)}_\infty \leq \kappa_A\max_{i\in\{0,1\}}{\abs{\lambda_i}^{n}\abs{1-\lambda_i^2}},
\end{equation}
where $\kappa_A=\norm{A}_\infty\norm{A^{-1}}_\infty$ is the condition number of $A$.
From the assumptions $\abs{\lambda_1}\geq \frac{9}{10}$, and $\abs{\theta-\frac{\pi}{2}}\leq \frac{\pi}{10}$, we establish that $\abs{\lambda_1^2-1}>1$, since $\Re[\lambda_1^2]=\abs{\lambda_1}^2\cos(2\theta)<0$, and also $\abs{\lambda_0\lambda_1-1}=\sqrt{1+\abs{\lambda_1}\lambda_0(\abs{\lambda_1}\lambda_0-2\cos(\theta))}>1$, since $2\cos(\theta)\leq 2\cos\left(\frac{2\pi}{5}\right)<\frac{18}{25}\leq \abs{\lambda_1}\lambda_0$.
Then, use \cref{lemma:condition_bound}, to put an upper-bound on $\kappa_A$
\begin{equation}\label{app:eq:kappa_A_lambda_min}
    \kappa_A\leq \sqrt{\frac{2}{2\lambda_{\min}^2-1}-1}\ ,
\end{equation}
where $\lambda_{\min}\coloneqq\min\Set{\lambda_0,\abs{\lambda_1}}$.
By the physicality conditions (\cref{app:PTM}), we have $\lambda_0\geq 2\abs{\lambda_1}-1$, and $\abs{\lambda_1}\leq 1$, which means that $\lambda_{\min}\geq 2\abs{\lambda_1}-1$.
In principle, we have all the ingredients to conclude the resulting bound in \cref{app:eq:abs_lambda_1_lb} by combining the inequalities in \cref{app:eq:abs_1-veps,app:eq:Ln(1-L),app:eq:kappa_A_lambda_min}.
However, we first need to exclude the case when the maximum on the right-hand side of \cref{app:eq:Ln(1-L)} corresponds to $i=0$.
First, notice that for $\lambda_{\min}\geq \frac{4}{5}$, we have $\kappa_A<\frac{5}{2}$, since the right-hand side of \cref{app:eq:kappa_A_lambda_min} is a monotonically decreasing function of $\lambda_{\min}$.
Then, for the case of $\abs{1-\veps_n}\geq 7\eta$, and $\abs{\lambda_1}\geq 1-\eta$, the conditions in \cref{app:eq:abs_1-veps,app:eq:Ln(1-L)} lead to a contradiction
\begin{equation}
7\eta\leq \abs{1-\veps_n}\leq \frac{1}{2}\kappa_A\lambda_0^{n}(1-\lambda_0^2)+2(1-\abs{\lambda_1})<5\lambda_0^{n}\abs{\lambda_1}(1-\abs{\lambda_1})+2(1-\abs{\lambda_1})\leq 7\eta,    
\end{equation}
where we again used the physicality constraints $\lambda_0\geq 2\abs{\lambda_1}-1$ and $\lambda_0\leq 1$, $\abs{\lambda_1}\leq 1$.

Finally, when combining the conditions in \cref{app:eq:abs_1-veps,app:eq:Ln(1-L),app:eq:kappa_A_lambda_min}, this time setting $i=1$ in \cref{app:eq:Ln(1-L)}, we use a trivial bound $\abs{1-\lambda_1^2}\leq 2$, and in the bound for $\kappa_A$, we insert $\lambda_{\min}\geq 2\abs{\lambda_1}-1$, obtaining
\begin{equation}\label{app:eq:tighter_abs_lambda_semifinal}
    \abs{1-\veps_n}\leq \abs{\lambda_1}^{n}\sqrt{\frac{2}{2(2\abs{\lambda_1}-1)^2-1}-1}+2(1-\abs{\lambda_1}),
\end{equation}
from where the bound in \cref{app:eq:abs_lambda_1_lb} follows.

Now we prove the simplified bound. Let us denote the function of $\abs{\lambda_1}$ for a fixed $n$ on the right-hand side of \cref{app:eq:tighter_abs_lambda_semifinal} by $F_n$, i.e., $\abs{1-\veps_n}\leq F_n(\abs{\lambda_1})$.
First, we establish that $\abs{\lambda_1}^{11}>0.7$.
Assume the opposite, and notice that $F_{22}$ is a monotonically increasing function on the interval $\abs{\lambda_1}\in [0.9,1]$, which is the range of $\abs{\lambda_1}$ considered in the lemma.
Therefore, if $\abs{\lambda_1}^{11}\leq 0.7$, then $F_{22}(\abs{\lambda_1})\leq F_{22}(0.7^{\frac{1}{11}})\approx 0.694<0.7$.
At the same time $0.7\leq \abs{1-\veps_n}\leq F_n(\abs{\lambda_1})\leq F_{22}(\abs{\lambda_1})$ for $n\geq 22$, leading to a contradiction.

Second, we assert that
\begin{equation}\label{app:eq:lemma_abs_lambda_proof_0.7}
    0.7\abs{\lambda_1}^{11}\sqrt{\frac{2}{2(2\abs{\lambda_1}-1)^2-1}-1}+2(1-\abs{\lambda_1})\leq 0.7,
\end{equation}
for $\abs{\lambda_1}^{11}>0.7$, which one can verify, e.g., by plotting the function on the left-hand side of \cref{app:eq:lemma_abs_lambda_proof_0.7}. This implies that $\abs{\lambda_1}^{n-11}\geq 0.7$ for $n\geq 22$, because otherwise
\begin{equation}
    F_n(\abs{\lambda_1})<0.7\abs{\lambda_1}^{11}\sqrt{\frac{2}{2(2\abs{\lambda_1}-1)^2-1}-1}+2(1-\abs{\lambda_1})\leq 0.7,
\end{equation}
contradicting our assumption $\abs{1-\veps_n}\geq 0.7$.

Having established that $\abs{\lambda_1}^{n-11}\geq 0.7$, and using the fact that $\abs{\lambda_1}^{11}\sqrt{\frac{2}{2(2\abs{\lambda_1}-1)^2-1}-1}\leq 1$, which follows from \cref{app:eq:lemma_abs_lambda_proof_0.7}, we finally get
\begin{equation}
    2(1-\abs{\lambda_1})\leq 0.7\left(1-\abs{\lambda_1}^{11}\sqrt{\frac{2}{2(2\abs{\lambda_1}-1)^2-1}-1}\right)\leq \abs{\lambda_1}^{n-11}\left(1-\abs{\lambda_1}^{11}\sqrt{\frac{2}{2(2\abs{\lambda_1}-1)^2-1}-1}\right),
\end{equation}
which leads to $F_n(\abs{\lambda_1})\leq \abs{\lambda_1}^{n-11}$ that proves our claim for the simplified bound.
\end{proof}

\section{Overestimation of gate fidelity for non-unitary gauge}\label{app:gauge}
Here we show that reporting quantum-gate fidelity only up to a non-unitary gauge can lead to an overestimation of the gate quality. The extent of this overestimation depends on the magnitude of the \ac{SPAM} errors, which is natural: a non-unitary gauge transformation can effectively reassign part of the gate error to the state preparation or measurement.

This can lead to a counterintuitive situation in which improving only the measurement device, while leaving both the state preparation and the implemented gate unchanged, results in a lower reported gate fidelity. Restricting the gauge freedom to unitary transformations avoids this conceptual shortcoming.

\begin{observation}
    Consider a family of noisy models $(\rho,\Set{\Lambda},\Set{M_0,M_1})$, with the following entries in the parametrization in \cref{app:eq:model_param}
    \begin{equation}
        L = \begin{pmatrix}
            a^2 & 0 & 0\\
            0 & 0 & -a\\
            0 & a & 0
        \end{pmatrix},\;
        \vec{l} = \begin{pmatrix}
            0 \\ 0 \\ 0
        \end{pmatrix},\; 
        \vec{r} = \begin{pmatrix}
            0 \\ 0 \\ 1
        \end{pmatrix},\; 
        \vec{m} = \begin{pmatrix}
            0 \\ 0 \\ b
        \end{pmatrix},\; 
        \mu = 0,
    \end{equation}
    where $a,b\in (0,1]$.
    In an experiment with observed probabilities $\Pr[0\vert n] = \Tr[\Lambda^n(\rho)M_0]$ for $n\in \Znn$, the average gate fidelity of $\Lambda$ can be overestimated due to optimization over a non-unitary gauge by at least
    \begin{equation}
        \frac{a}{6}\min\Set*{\frac{(1-a)^2}{a},\frac{(1-b)^2}{b}},
    \end{equation}
    which for $b>a$ depends on the measurement quality, and decreases with increasing $b$.
\end{observation}
\begin{proof}
    First, let us calculate the true average gate fidelity between $\Lambda$ and $\sqrt{\Xgate}$-gate, under a unitary gauge $U$, which we represent by the orthogonal transformation $V\in \OO(3)$ acting on $L_0$, 
    \begin{equation}
        \Fid_\avg(\sqrt{\Xgate}_{\vert U},\Lambda) = \frac{1}{2}+\frac{1}{6}\Tr[VL_0^\T V^\T L] \leq \frac{1}{2}+\frac{1}{6}\norm{L}_1=\frac{1}{3}+\frac{(1+a)^2}{6},
    \end{equation}
    with the above inequality saturated by e.g., $V=\1$. Let us now consider a matrix $P$
\begin{equation}
    P = \begin{pmatrix}
        1 & 0 & 0\\
        0 & 1 & 0\\
        0 & 0 & c
    \end{pmatrix},
\end{equation}
for $c\in (0,1]$, which is, clearly, invertible.
The current \ac{QCVV} literature (see e.g., Ref.~\cite{blumekohout2025qcvv}) suggests considering the physical models connected by arbitrary invertible transformations as equivalent, which, for the considered case, would mean that one can map $L\to PLP^{-1}$, $\vec{r}\to P\vec{r}$, and $\vec{m}^\T\to \vec{m}^\T P^{-1}$, which indeed preserves the measurement statistics,
\begin{equation}
    \frac{1}{2}+\frac{1}{2}\vec{m}^\T P^{-1} (P L P^{-1})^n P\vec{r} = \frac{1}{2}+\frac{1}{2}\vec{m}^\T L^n \vec{r} = \Pr[0\vert n].
\end{equation}
Clearly, for $c\in(0,1]$, the Bloch vector $P\vec{r}$ corresponds to a physical state.
The measurement remains physical whenever $c\geq b$.
The Choi-Jamio\l{}kowski ``state'' of the channel in the gauge $P$ is 
\begin{equation}
    \frac{1}{4}\begin{pmatrix}
        1 & \i\frac{a}{c} & \i ac & a^2\\
        -\i\frac{a}{c} & 1 & a^2 & -\i ac\\
        -\i ac & a^2 & 1 & -\i\frac{a}{c}\\
        a^2 & \i ac & \i\frac{a}{c} & 1
    \end{pmatrix},
\end{equation}
and for $a,c\in(0,1]$ its positivity is equivalent to the condition $c\geq a$.
Finally, the formula for the fidelity in the gauge $P$ becomes $\frac{1}{2}+\frac{a}{6}\left(a+c+\frac{1}{c}\right)$, which is greater than the true value $\frac{1}{2}+\frac{a}{6}(a+2)$ for any $c\in (0,1]$ and the difference increases as $c$ decreases. 
Therefore, we can conclude that overestimation of the gate fidelity (or equivalently, underestimation of the infidelity) due to non-unitary gauge optimization is at least
\begin{equation}
    \frac{a(1-c)^2}{6c},
\end{equation}
where $c\geq\max\Set{b,a}$.
First of all, we see that whenever $b<1$ and $a< 1$, i.e., the measurement and gate are not perfect, the average gate fidelity is overestimated due to the non-unitary gauge.
Moreover, if $b>a$, improvement in the measurement can effectively decrease the estimate for the gate fidelity.
\end{proof}

\section{Technical lemmata}\label{app:lemmata}
In this section, we give some technical lemmata, which are mathematical facts used in the proofs of the main results.

\begin{lemma}\label{lemma:phi}
    Let $\abs{1+r\e^{\ii\theta}}\leq \xi$ for $\theta\in[0,2\pi]$, $r\in [0,1]$ and $\xi\in[0,1)$.
    Then $\abs{\sin(\theta)}\leq \xi$.
\end{lemma}
\begin{proof}
First, we treat the case $\xi=0$.
In this case we have $r\e^{\i\theta}=-1$, which immediately forces $\theta = k\pi$ for $k$ odd integer, which then implies $\sin(\theta)=0$.

We prove the bound $\sin(\theta)\leq \xi$ for $\xi\in(0,1)$ by applying the \ac{KKT} condition to the following optimization problem
\begin{align}\label{app:eq:opt_phi}
    \max_{\theta,r}\quad & \sin(\theta)\\
     \text{s.t.}\quad & (1+r\cos(\theta))^2+(r\sin(\theta))^2\leq \xi^2,\nonumber \\
     & 0\leq r\leq 1,\nonumber
\end{align}
where the boundary conditions for $\theta$ are not considered, since the cases $\theta=0$ and $\theta=2\pi$ are clearly unfeasible. 
The Lagrangian corresponding to the optimization problem~\eqref{app:eq:opt_phi} is
\begin{equation}
    \mathcal{L} = -\sin(\theta)+\mu_0\left((1+r\cos(\theta))^2+(r\sin(\theta))^2-\xi^2\right)+\mu_1(r-1)-\mu_2 r,
\end{equation}
where $\mu_0,\mu_1,\mu_2\in\Rnn$ are the dual variables.
The stationarity conditions read
\begin{equation}\label{app:eq:stat_cond}
    \begin{split}
        \frac{\partial \mathcal{L}}{\partial \theta} & = -\cos(\theta)-2r\mu_0\sin(\theta)=0,\\
        \frac{\partial \mathcal{L}}{\partial r} & = 2\mu_0(\cos(\theta)+r)+\mu_1-\mu_2=0, 
    \end{split}
\end{equation}
and the complementary slackness reads
\begin{equation}
    \mu_0(\abs{1+r\e^{\ii\theta}}^2-\xi^2)=0,\quad \mu_1(r-1)=0,\quad \mu_2r=0.
\end{equation}
We can directly exclude the case $r=0$, because it is incompatible with the constraint $\abs{1+r\e^{\ii\theta}}\leq \xi<1$.
Therefore, it must be that $\mu_2=0$.
Similarly, if we assume $\mu_0=0$, from the first stationarity condition in \cref{app:eq:stat_cond}, we obtain $\cos(\theta)=0$ (and hence $\sin(\theta)=\pm 1$), which is again not compatible with $\abs{1+r\e^{\ii\theta}}\leq \xi$, since $\abs{1\pm\ii r}\geq 1$. 
Therefore, the maximum in \cref{app:eq:opt_phi} is attained at $\theta$ and $r$, for which $\abs{1+r\e^{\ii\theta}}=\xi$.
Next, if we assume $r=1$, then the second stationarity condition in \cref{app:eq:stat_cond} becomes $2\mu_0(\cos(\theta)+1)+\mu_1=0$, which cannot be satisfied for $\mu_0\geq 0$ and $\mu_1>0$.
Therefore, $\mu_1=0$. Having settled the complementary slackness, we resolve the remaining set of conditions
\begin{equation}
    \begin{split}
        \cos(\theta)+r=0,\quad 
        \cos(\theta)+2r\mu_0\sin(\theta)=0,\quad       \abs{1+r\e^{\ii\theta}} &=\xi,
    \end{split}
\end{equation}
which lead to $\sin(\theta)=\xi$, $r=\sqrt{1-\xi^2}$, and $\mu_0=\frac{1}{2\xi}$, corresponding to the global optimum.
The bound $\sin(\theta)\geq -\xi$ can be obtained from \cref{app:eq:opt_phi} by changing the sign of $\theta$.
\end{proof}

\begin{lemma}\label{lemma:non-unitality_bound} For the parametrization in \cref{app:eq:model_param} of a qubit channel $\Lambda$, assume that $L$ is invertible and that eigenvalues of $L$ are $\Set{\lambda_0,\lambda_1,\lambda_1^\ast}$, where $\lambda_0 \in \RR_{\geq 0}$ and $\lambda_1\in\CC\setminus\RR$. Then
\begin{equation}
\norm{\vec l} \leq 2(1-\abs{\lambda_1}).
\end{equation}

\end{lemma}
\begin{proof}
    Positivity of the channel implies that for any normalized $\vec{v}$ we have $\norm{\vec{l}+L\vec{v}}\leq 1$. 
    Then taking $\vec{v}=L^{-1}\vec{l}/\norm{L^{-1}\vec{l}}$ leads to
\begin{align}\label{app:eq:lemma_l}
    1\geq \norm[\Bigg]{\vec{l}+L\frac{L^{-1}\vec{l}}{\norm{L^{-1}\vec{l}}}}=\left(1+\frac{1}{\norm{L^{-1}\vec{l}}}\right)\norm{\vec{l}}\geq\left(1+\frac{1}{\norm{L^{-1}}_\infty\norm{\vec{l}}}\right)\norm{\vec{l}}=\norm{\vec{l}}+\left(\norm{L^{-1}}_\infty\right)^{-1} = \norm{\vec{l}}+\sigma_{3},
\end{align}
where $0 < \sigma_{3}\leq \sigma_{2}\leq \sigma_{1}$ are the singular values of $L$. From the complete positivity, we have the constraint~\cite{wolf2010inverse}
\begin{equation}\label{app:eq:lemma_l_sigmas}
    \sigma_1+\sigma_2\leq 1+\sigma_3.
\end{equation}
Let $\mu_1\geq \mu_2\geq \mu_3$ be ordered absolute values of the eigenvalues of $L$, then we have $\mu_1+\mu_2\geq 2\abs{\lambda_1}$. Using Weyl's theorem~\cite{Weyl1949}, we have 
\begin{equation}
    \sigma_1+\sigma_2 \geq \mu_1+\mu_2\geq 2\abs{\lambda_1}.
\end{equation}
Combining with \cref{app:eq:lemma_l_sigmas}, we immediately get $\sigma_3\geq 2\abs{\lambda_1}-1$. 
Inserting it into \cref{app:eq:lemma_l} proves the claim.
\end{proof}

\begin{lemma}\label{lemma:condition_bound}
    Let $L\in \CC^{3\times 3}$, such that $\norm{L}_\infty\leq 1$ be a matrix with distinct eigenvalues $\lambda_1, \lambda_2, \lambda_3\in \CC$ satisfying $1\geq \abs{\lambda_1}\geq \abs{\lambda_2}\geq \abs{\lambda_3}>\sqrt{1-\frac{\delta}{2}}$ and $\abs{\lambda^\ast_i\lambda_j -1}\geq \delta$ for all $i\neq j\in \Set{1,2,3}$, and $\delta\in (0,2]$. 
    Then there exists $A$ invertible, such that $L=A \diag(\lambda_1,\lambda_2, \lambda_3) A^{-1}$ with its condition number satisfying
    \begin{equation}
        \kappa_A\leq\sqrt\frac{\delta+2(1-\abs{\lambda_3}^2)}{\delta-2(1-\abs{\lambda_3}^2)}\ .
    \end{equation}
\end{lemma}

\begin{proof} Let $A=\begin{pmatrix} \vec v_1 & \vec  v_2 & \vec v_3\end{pmatrix}$, where the columns are the eigenvectors of $L$, normalized but not necessarily orthogonal. We first show that for all $i\neq j\in \Set{1,2,3}$,
\begin{equation}
    \delta \abs{\vec{v_i}^\dagger \vec{v_j}}\leq {1-\abs{\lambda_3}^2}.
\end{equation}
To see this, consider the following chain of inequalities,
\begin{align}
    \begin{split}
        \delta\abs{\vec{v_i}^\dagger \vec{v_j}} \leq \abs{(1-\lambda_i^\ast \lambda_j)\vec{v_i}^\dagger \vec{v_j}}= \abs{\vec{v_i}^\dagger (\1-L^\dagger L) \vec{v_j}}\leq \sqrt{\vec{v_i}^\dagger (\1-L^\dagger L) \vec{v_i}}\sqrt{\vec{v_j}^\dagger (\1-L^\dagger L) \vec{v_j}}\leq 1-\abs{\lambda_3}^2,
    \end{split}
\end{align}
where the second inequality above follows from the fact that $\1-L^\dagger L$ is positive-semidefinite (since $\norm{L}_\infty\leq 1$).
Next, for any $\vec {x}=\begin{pmatrix} x_1 & x_2 & x_3\end{pmatrix}^\T \in \CC^3$ with $\norm{\vec{x}}=1$, we have
\begin{equation}
    \norm{A\vec x}^2=\norm{\vec{x}}^2+\sum_{i\neq j}x_i^\ast{x}_{j} \vec{v_i}^\dagger\vec {v_{j}}\leq 1+\frac{2(1-\abs{\lambda_{3}}^2)}{\delta}\sum_{i<j} \abs{x_i^\ast x_{j}}\leq 1+\frac{2(1-\abs{\lambda_{3}}^2)}{\delta}.
\end{equation}
Similarly, we can show that 
\begin{equation}
 \norm{A\vec x}^2\geq  1-\frac{2(1-\abs{\lambda_{3}}^2)}{\delta},
\end{equation}
from which the upper bound on the condition number, as stated in the Lemma, follows:
\begin{equation}
    \kappa_A = \norm{A}_\infty\norm{A^{-1}}_\infty=\frac{\max_{\norm{\vec{x}}=1}\norm{A\vec{x}}}{\min_{\norm{\vec{x}}=1}\norm{A\vec{x}}}\leq \sqrt{\frac{\delta+2(1-\abs{\lambda_3}^2)}{\delta-2(1-\abs{\lambda_3}^2)}}\ .
\end{equation}
 \end{proof}

\end{appendix}
\twocolumngrid
\bibliographystyle{unsrt}
\bibliography{ref.bib}
\end{document}